\PassOptionsToPackage{table,xcdraw}{xcolor}
\documentclass[sigplan,10pt]{acmart}
\AtBeginDocument{%
  }

\usepackage{booktabs}
\usepackage{tabularx}
\usepackage{array}

\usepackage{soul}
\setul{0.5ex}{0.1ex} 
\usepackage{graphicx}    
\usepackage{subcaption}  
\usepackage[ruled,vlined]{algorithm2e}
\usepackage{array}
\usepackage{xspace}
\usepackage{pifont}
\usepackage{multirow}
\usepackage{tikz}  
\usepackage{amsmath}
\ifdefined\symbf
  \let\bm\symbf
\else
  \usepackage{bm}
\fi
\newcommand{\squishlist}{
	\begin{list}{$\bullet$}
		{ \setlength{\itemsep}{2pt}      \setlength{\parsep}{3pt}
			\setlength{\topsep}{3pt}       \setlength{\partopsep}{0pt}
			\setlength{\leftmargin}{5.5mm} \setlength{\labelwidth}{1em}
			\setlength{\labelsep}{0.5em} } }
	\newcommand{\squishend}{
\end{list}  }

\usepackage[table]{xcolor}
\usepackage{multirow}
\usepackage{makecell}
\usepackage{booktabs}
\usepackage{threeparttable} 
\usepackage{adjustbox}     
\definecolor{rowA}{RGB}{245,248,255}
\definecolor{rowB}{RGB}{250,250,250}
\definecolor{rowC}{RGB}{248,250,245}  
\definecolor{rowD}{RGB}{255,248,245}  
\definecolor{rowE}{RGB}{248,245,255}  

\newcommand{\cellA}{\cellcolor{rowA}}
\newcommand{\cellB}{\cellcolor{rowB}}

\usepackage{tikz}  
\usepackage{xcolor}

\SetKwComment{Comment}{$\triangleright$ }{}
\usepackage[normalem]{ulem}
\usepackage{graphicx}
\usepackage{minted} 
\usepackage{listings}
\usepackage{subcaption}
\usepackage{caption}
\usepackage{enumerate}
\usepackage{enumitem}
\usepackage{multirow}
\newtheorem{remark}{Remark}

\newtheorem{theorem}{Theorem}
\newtheorem{example}{Example}
\newcommand{\para}[1]{\noindent {\bf #1} \hspace{2pt}}

\usepackage{enumerate}
\usepackage{enumitem}
\usepackage{hhline}

\newcommand{\eg}{\emph{e.g.,}\xspace}

\usepackage{makecell}
\DeclareMathOperator*{\argmin}{arg\,min}
\DeclareMathOperator*{\argmax}{arg\,max}
\newcommand{\boxedref}[2]{%
  \hyperref[#1]{%
    \begingroup
      \setlength{\fboxsep}{0.5pt}% 内边距，越小越“贴字”，可改成 0pt 试试
      \setlength{\fboxrule}{0.5pt}% 线条粗细
      \color{green}% 只给框上色
      \fbox{\color{black}#2}% 框里文字重新设回黑色
    \endgroup
  }%
}

\newcommand{\codehl}[1]{\textcolor{green!60!black}{#1}}
\newcommand{\codehll}[1]{\textcolor{blue!70!black}{#1}}

\newcommand{\ByteX}{\textsf{UBASE}\xspace}

\newcommand{\AuthorSep}{,\enspace}

\usepackage{scalefnt}

\usepackage{caption}
\newcommand{\codecomment}[1]{\textcolor{gray!60}{#1}}
\setminted{
  fontsize=\footnotesize,
  breaklines,
  frame=single,
  framesep=2mm,
  linenos=false,
  escapeinside=||
}

\newif\iftechreport
\newif\ifrelease

\releasetrue      

\techreportfalse   

\copyrightyear{2026}
\acmYear{2026}
\setcopyright{cc}
\setcctype{by-nc-nd}
\acmConference[SIGOPS ATC 2026]{The 2026 ACM SIGOPS Annual Technical Conference}{November 15--18, 2026}{Hong Kong, China}
\acmBooktitle{The 2026 ACM SIGOPS Annual Technical Conference (SIGOPS ATC 2026), November 15--18, 2026, Hong Kong, China}
\acmDOI{10.1145/3828523.3857093}
\acmISBN{979-8-4007-2818-1/2026/11}

\begin{document}
\pagestyle{plain}
\title{UBASE: An AI Search Engine for Trillion-Scale Vector Data Management at ByteDance}

\author{%
\begin{tabular}{c}
Yao Tian\textsuperscript{\textdagger}\AuthorSep
Yuncheng Lu\textsuperscript{\textdagger}\AuthorSep
Liyao Xiong\AuthorSep
Yuming Xu\AuthorSep
Hao Zhang\textsuperscript{*}\AuthorSep
Weichen Zhao
\tabularnewline
Xi Zhao\AuthorSep
Bo Kuang\AuthorSep
Dongyu Wang\AuthorSep
Jiehui Li\AuthorSep
Yakun Li\textsuperscript{*}\AuthorSep
Lei Zhang
\tabularnewline[0.25em]
ByteDance
\end{tabular}%
} 
\renewcommand{\shortauthors}{Tian et al.}

\begin{abstract}
Since 2016, \ByteX has been the foundation of ByteDance's search infrastructure, scaling to more than 7,000 clusters and 300 PB of indexed data. Driven by the demands of AI workloads, \ByteX has evolved from a text search engine into a unified AI search system supporting vector retrieval, lexical matching, and predicate filtering. Its largest deployment indexes nearly one trillion high-dimensional vectors. This scale exposes two central bottlenecks in AI-era retrieval: memory-intensive graph-index construction under sustained ingestion, and the prohibitive cost of keeping vector indexes entirely in memory.

\ByteX addresses these bottlenecks with two techniques. First, it introduces a quantization-aware vector kernel based on \textsf{SymRaBitQ}, a new symmetric quantization scheme with theoretical analysis that allows index construction to run directly in the quantized space accurately and efficiently without retaining a copy of full-precision vectors. Second, it provides a hybrid storage engine that supports memory-resident, hybrid, and SSD-resident deployments, with fine-grained record-level caching to trade memory for latency under operational control. On large-scale benchmarks, \ByteX improves throughput by up to $3\times$, reduces indexing memory by $80\%$, and lowers serving cost (\$/QPS) by up to $99.4\%$ compared to prior systems, while supporting trillion-vector scale, write-heavy or latency-sensitive workloads in production.
\end{abstract}





\maketitle
\begingroup
\renewcommand{\thefootnote}{\fnsymbol{footnote}}

\footnotetext[2]{Equal contribution:
\href{mailto:yao.tian@bytedance.com, luyuncheng@bytedance.com}
{\{yao.tian, luyuncheng\}@bytedance.com}.}

\footnotetext[1]{Corresponding authors:
\href{mailto:zhanghao.ai@bytedance.com, liyakun.hit@bytedance.com}
{\{zhanghao.ai, liyakun.hit\}@bytedance.com}.}

\endgroup

\section{Introduction}
\label{sec:intro}

\ByteX{}\footnotemark{} has served as ByteDance's core search infrastructure since 2016. Today, \ByteX{} manages more than 7,000 clusters and over 300 PB of storage. AI applications, such as retrieval-augmented generation~\cite{lewis2020retrieval, DBLP:journals/corr/abs-2608-01565}, recommendation~\cite{li2018design, peng2024large}, semantic analytics~\cite{patel2025semantic, liu2024suql}, and conversational assistants~\cite{githubcopilot2023, cursor2024, DBLP:journals/corr/abs-2607-29440}, have driven \ByteX to evolve from text-based search engine to \emph{AI search} engine, where vectors are first-class indexed data and queries combine vector similarity with lexical matching, structured predicates, and reranking.

Vector data growth is multiplicative rather than incremental. As shown in Fig.~\ref{fig:overalltrend}, \ByteX is storing vectors at an exponential rate, while Fig.~\ref{fig:dailygrowth} shows two production workloads each ingesting several billion new vectors per day. Our largest clusters are already nearing \emph{trillion-vector} scale with PB-scale storage footprints. The demand is rising along three coupled axes: granularity, dimensionality, and retention. Retrieval units are becoming finer-grained, \eg video search moving from scene-level to frame-level embeddings; newer encoders are increasing embedding dimensionality from 128 to 2048; and vectors are being stored for longer periods, e.g., AI-memory workloads extending retention from 45 to 365~days. These trends compound: each application may generate more vectors per user, store wider vectors, and retain them for longer. Consequently, vector data management has shifted from a component-level optimization to a PB-scale system-level optimization.

\begin{figure*}[t]
    \centering
    \begin{minipage}[t]{1.72in}
        \centering
        \includegraphics[width=\linewidth]{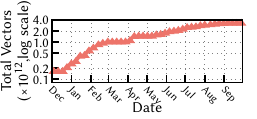}
        \par
    \parbox{0.85\linewidth}{%
        \caption{Overall Growth: from 2025-12 to 2026-09, one point per week}
        \label{fig:overalltrend}
    }
    \end{minipage}%
    \hfill
    \begin{minipage}[t]{1.72in}
        \centering
        \includegraphics[width=\linewidth]{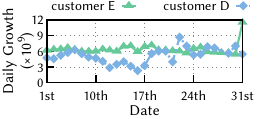}
        \par
        \parbox{0.9\linewidth}{%
        \caption{Daily Growth: from 2026-05-01 to 2026-06-01, one point per day}
        \label{fig:dailygrowth}
    }
    \end{minipage}%
    \hfill
    \begin{minipage}[t]{1.72in}
        \centering
        \includegraphics[width=\linewidth]{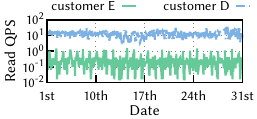}
        \par
        \parbox{0.9\linewidth}{%
        \caption{Daily QPS: from 2026-05-01 00:00 to 2026-06-01 00:00, one point per hour}
        \label{fig:dailyQPS}
    }
    \end{minipage}
    \hfill
    \begin{minipage}[t]{1.72in}
        \centering
        \includegraphics[width=\linewidth]{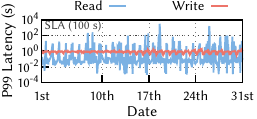}
        \par
        \parbox{\linewidth}{%
        \caption{P99 of Customer E: from 2026-05-01 00:00 to 2026-06-01 00:00, one point per hour}
        \label{fig:dailyLatency}
    }
    \end{minipage}
    \vspace{-1em}
\end{figure*}
\footnotetext{\ByteX{} originated from an early OpenSearch branch; we thank the community for its foundational contributions.}

Vector data is resource-intensive and challenging to manage. Unlike scalar fields, each vector is a dense object that is expensive to store, index, and query. A 2,048-dimensional FP32 vector occupies 8~KB, so 100 million vectors require 800~GB of raw storage before indexing. Vectors are usually partitioned into random segments during ingestion\footnote{Production vector engines such as OpenSearch/ElasticSearch~\cite{opensearch2024, elasticsearch2024} and Milvus~\cite{DBLP:conf/sigmod/WangYGJXLWGLXYY21} commonly adopt segment-based index management, which supports high-throughput vector ingestion while simplifying incremental addition and deletion in deployed systems.}, and indexes are built over them. Index construction for these vectors is memory- and compute-intensive. For example, building a graph-based vector index \cite{weaviate-hnsw-vamana, DBLP:conf/nips/SubramanyaDSKK19}, typically the fastest option for vector search, for a 1-million-vector segment requires at least 8~GB of memory and takes over one hour. Querying, in other words, finding the approximate nearest neighbor (ANN) vectors to a query vector, is similarly expensive: graph-based methods employ a greedy traversal to progressively move toward vertices closer to the query. Each hop takes at least 1 random I/O, and there could be 100+ hops to reach 95\% recall in a single segment, and it must search all segments.

At trillion-vector scale, with ingestion and serving running concurrently, memory becomes the central constraint: \textit{every byte spent on index construction is a byte unavailable for serving}. Our production workloads in \S\ref{sec:preliminary} expose two bottlenecks. First, ingestion amplifies memory usage: index builders must retain full-precision vectors until segment finalization, so at 10+ billion vectors ingested per day and peaks above 700K vectors/s, naive construction can consume tens of terabytes and exhaust cluster memory. Second, serving becomes I/O-bound because memory-only serving is economically infeasible, making SSD-backed serving necessary at a trillion-vector scale\footnote{VM memory costs roughly 67$\times$ more per byte than SSD storage. A typical VM instance with 8\,GB memory is priced at $\approx$ 6.7\$/GB-month, versus an SSD-backed block device is priced at 0.10\$/GB-month.}. Cache efficiency is therefore critical: misses on graph neighbors or vectors trigger SSD reads, increasing tail latency and reducing throughput.

These problems cannot be avoided by moving vector search into an isolated service. Customers require OpenSearch compatibility to preserve existing APIs, dashboards, access control, ranking pipelines, and operational playbooks,\footnote{Many customers are migrating from OpenSearch-based lexical search to hybrid search that combines vector, lexical, and structured retrieval; preserving OpenSearch compatibility is therefore essential.} while also demanding hybrid queries that combine vector similarity, lexical matching, structured predicates, and reranking. Existing systems miss one side of this requirement. Faiss~\cite{faiss} provides efficient ANN kernels but lacks distributed ingestion, durability, and operational control. Milvus~\cite{DBLP:conf/sigmod/WangYGJXLWGLXYY21} provides an integrated vector service but does not match OpenSearch's mature lexical and ranking ecosystem. OpenSearch and Elasticsearch~\cite{opensearch2024, elasticsearch2024} provide the right ecosystem, but their vector paths hit memory and I/O bottlenecks at trillion-vector scale. This leaves a practical gap: production hybrid search needs OpenSearch compatibility and extreme-scale vector performance in the same system.

We present \ByteX, a production-grade search engine for AI workloads built inside the OpenSearch/Elasticsearch ecosystem. Rather than proposing a standalone vector service, \ByteX replaces the quantization, storage, and query-execution paths that bottleneck at extreme scale, enhances hybrid search, enables unified storage across memory-only, hybrid, and disk-only deployments, and preserves existing APIs, DSLs, dashboards, and operational playbooks.  Its design follows a deployment constraint: at trillion-vector scale, the hard problem is not only low-latency ANN search, but sustaining ingestion, index, storage, and hybrid query execution within an operational search platform.

\ByteX addresses these challenges through two design choices. First, compression is integrated into the indexing path. Full-precision vectors otherwise dominate memory consumption during index construction, so \ByteX introduces \textsf{SymRaBitQ} (\S\ref{sec:vector_kernel}), a symmetric quantization scheme with solid theoretical guarantees. Its high-precision, error-bounded distance estimator operates directly on quantized codes, eliminating full-precision vector copies and substantially reducing indexing cost.
Second, storage placement is adaptive and optimized for cache efficiency. Production clusters span memory-rich, memory-constrained, and SSD-heavy deployments, and their cost--latency trade-offs change with workload. \ByteX supports memory-resident, hybrid memory-disk, and SSD-resident execution within the same system, using fine-grained record-level caching and asynchronous graph traversal optimized for SSD I/O patterns (\S\ref{sec:storage_engine}). Both design choices are integrated into an OpenSearch-compatible architecture that preserves support for rich predicates, ranking functions, and hybrid retrieval workloads (\S\ref{sec:query_engine}).
Together, these choices make \ByteX a drop-in evolution of existing search infrastructure rather than a parallel vector stack, closing the gap between ecosystem maturity and trillion-scale AI retrieval without requiring migration or abandoning operational tooling.
On large-scale benchmarks, \ByteX delivers up to $3\times$ higher throughput, $80\%$ lower indexing memory, and $99.4\%$ lower operating cost compared to state-of-the-art systems while preserving recall and latency.  In production, \ByteX supports both PB-scale write-heavy clusters and latency-sensitive online serving clusters within the same operational ecosystem.

\begin{table*}[t]
\centering
\scriptsize
\setlength{\tabcolsep}{2.6pt}
\renewcommand{\arraystretch}{1.12}
\caption{Representative \ByteX{} production workloads across query patterns, ingestion rates, storage tiers, and scale.}
\label{tab:prod-workloads}
\begin{threeparttable}
\begin{tabularx}{\textwidth}{@{}
>{\raggedright\arraybackslash}p{0.28cm}
>{\raggedright\arraybackslash}p{3.00cm}
>{\raggedright\arraybackslash}X
>{\raggedleft\arraybackslash}p{1.5cm}
>{\raggedleft\arraybackslash}p{0.48cm}
>{\raggedleft\arraybackslash}p{0.76cm}
>{\raggedleft\arraybackslash}p{0.70cm}
>{\raggedright\arraybackslash}p{0.82cm}
>{\raggedright\arraybackslash}p{0.70cm}
>{\raggedleft\arraybackslash}p{0.82cm}
@{}}
\toprule
\textbf{ID} &
\textbf{Workload} &
\textbf{Access pattern} &
\textbf{Vectors} &
\textbf{Dim} &
\textbf{Write/s} &
\textbf{QPS} &
\textbf{SLA} &
\textbf{Tier} &
\textbf{Size} \\
\midrule
A &
Image Assets Search &
Text-to-Image retrieval with tenant ID, tags, and time ordering &
125\,M & 768 & 100 & 10+ & <100\,ms & Memory & 2\,TB \\

B &
Multimodal Assets Search &
Multimodal-to-Multimodal\tnote{2} retrieval with tenant ID and keywords&
52\,B & 1024 & 65 & 1000+ & <100\,ms & Memory & 120\,TB \\

C &
AD Recommendation &
Video-to-Video retrieval with range predicates &
28.4\,B & 512 & 81\,K & 10+ & <1\,s & Memory & 212\,TB \\

D &
Compliance Review Search &
(Text, Image)-to-Image retrieval with tags &
1.17\,T & 2048 & 100+\,K\tnote{3} & 1--10 & <100\,s & SSD & 5.1\,PB \\

E &
User Identity Search &
(Text, Image)-to-(Text, Image) retrieval with time range and tags &
1.07\,T & 512 & 100+\,K & 1--10 & <100\,s & SSD & 3.79\,PB \\

F &
Pretraining Data Search &
Text-to-Text retrieval with content filters &
52\,B & 1024 & offline\tnote{1} & 100\,K+ & <100\,ms & Memory & 100\,TB \\
\bottomrule
\end{tabularx}
\begin{tablenotes}[flushleft]
\scriptsize
\item[1] Index is constructed offline in batch, no online writes.
\item[2] Multimodal: Text, Image, Video, Audio.
\item[3] Peak at 700+\,K
\end{tablenotes}
\end{threeparttable}
\vspace{-1.5em}
\end{table*}

In summary, this paper makes the following contributions:
\begin{itemize}[leftmargin=*, nosep]
  \item We distill the production insights from ByteDance's AI search platform--7\,000+ clusters, 300+\,PB storage, trillion-scale vectors--identifying memory-efficient graph construction during ingestion and vector storage during serving as the two dominant physical bottlenecks at extreme scale.
  \item We introduce \textsf{SymRaBitQ}, a quantization scheme with strong theoretical accuracy guarantees, and a quantization-aware vector kernel that enables search, graph construction, and segment merge entirely in the quantized space.
  \item We design an adaptive storage engine supporting memory-resident, hybrid, and SSD-resident vector indexes with fine-grained cache management for efficiently utilizing cache in trillion-scale vector workloads, and a hybrid query engine that integrates vector, lexical, predicate, and reranking operators within a single distributed query plan.
  \item We evaluate \ByteX on public benchmarks and production deployments at scale, up to 1.17 trillion vectors and 5.1\,PB of storage, achieving up to $80\%$ lower graph-build peak memory, $60\%$ faster build time, and $99.4\%$ lower serving cost while preserving recall and latency.
\end{itemize}

\section{Production Workload \& Challenges}
\label{sec:preliminary}

\subsection{Production Workloads}
\label{sec:workloads}
Table~\ref{tab:prod-workloads} summarizes representative \ByteX{} workloads across 7000+ clusters/services with over 300\,PB storage. Customers first adopted \ByteX{} as an OpenSearch-compatible platform, then incrementally enabled vector search over existing queries and workflows. Consequently, \ByteX{} cannot expose vector search as a standalone ANN service: it must preserve query semantics, shard/segment management, filtering, ranking, and hybrid retrieval across memory- and SSD-resident indices. Workloads expose overlapping features:
\begin{itemize}[leftmargin=*, nosep]
\item \textbf{Hybrid queries.}
Customers~A--F combine vector similarity with lexical, scalar, and domain-specific predicates, then apply ranking and LLM-based recall. Their cost is therefore dominated not by ANN traversal alone, but by materialization, filtering, merging, and downstream reranking. These workloads require end-to-end hybrid query optimization.

\item \textbf{Low-latency online serving.}
Memory-resident workloads, such as A, B, C, and F, target P99 latencies in the tens to hundreds of milliseconds for high-dimensional vector queries. 

\item \textbf{PB-scale ingestion \& serving.}
Disk-resident workloads, such as D and E, operate at PB scale, where the memory-only option is economically infeasible. As shown in Fig.~\ref{fig:dailyQPS} and Fig.~\ref{fig:dailyLatency}, they have relatively low query demand and comparatively relaxed latency requirements: 1–10 QPS and a 100s latency SLA. The primary challenge is sustaining ingestion and serving over massive datasets.
\end{itemize}





\para{System requirements.}
These workloads impose three requirements on \ByteX{}.

\begin{itemize}[leftmargin=1.2em,itemsep=0.25em,topsep=0.25em,parsep=0pt,partopsep=0pt]
\item \textbf{OpenSearch-native hybrid execution.}
\ByteX{} must preserve the OpenSearch compatibility---including ingestion, segment merge, replication, recovery, monitoring, query DSL, filtering, ranking, and shard-level execution---while making vector search a first-class operator inside the same distributed query plan. 

\item \textbf{Unified memory--SSD backend.}
\ByteX{} must support memory-resident, SSD-resident, and mixed deployments with one architecture. As Table~\ref{tab:prod-workloads} shows, production deployments span 2--212\,TB memory-backed workloads and 3--5\,PB SSD-backed workloads, so the system cannot rely on a DRAM-only ANN design or a separate disk-specialized serving path.

\item \textbf{Large-scale continuous ingestion and serving.}
\ByteX{} must sustain hundreds of billions to trillions of vectors under continuous writes while serving queries in reasonable times. At this scale, ingestion-side index construction and serving-side cache management contend for the same memory, making ingestion and serving a coupled system problem rather than two independent pipelines.

\end{itemize}

\subsection{Challenges}
\label{sec:challenges}

At the trillion-vector scale, DRAM is the shared bottleneck between ingestion and serving. In practice, naïve graph construction can transiently consume more memory than the final index itself. For Customer~E, construction peaked at 27\,TB against a 13.8\,TB DRAM budget, making the system infeasible even before reserving memory for the serving-side caches needed to hide SSD latency. This exposes a fundamental trade-off: allocating DRAM for fast ingestion inflates peak memory, while leaving too little for caching slows query throughput. Two interdependent challenges emerge:

\begin{itemize}[leftmargin=*, nosep]

\item \textbf{Excessive peak memory during graph construction.} Conventional builders retain full-precision vectors, candidate neighborhoods, pruning state, partial edges, and metadata until segments are finalized, inflating peak memory. The challenge is reducing this peak without fragmenting segments or degrading index quality, ensuring that SSD traversal remains efficient downstream.

\item \textbf{Ineffective DRAM use for SSD-resident graph search.}
Trillion-scale indices cannot fit in DRAM, so queries traverse SSD-resident graphs, vectors, and metadata. The challenge is that DRAM capacity alone does not determine performance: cached data only helps when it avoids SSD accesses on the query-critical path. Without traversal-aware cache allocation, scarce DRAM can be spent on bytes with low latency impact, while hot graph regions still trigger SSD reads.

\end{itemize}

\para{Limitations of existing systems.}
Disk-based ANN in OpenSearch and Elasticsearch~\cite{opensearch2024,elasticsearch2024} reduces serving-time DRAM with better binary quantization (BBQ) \cite{bbq}, but leaves the ingestion--serving memory conflict unresolved. On the ingestion side, BBQ is asymmetric: its compressed representation is used for search, but graph construction still materializes full-precision vectors and construction metadata, directly competing with serving caches for DRAM. On the serving side, generic \texttt{mmap}/page caches operate at page granularity and cannot prioritize ANN records on query-critical traversal paths. \ByteX{} addresses this shared bottleneck by quantized construction (\S\ref{sec:vector_kernel}), and a fine-grained DRAM--SSD cache (\S\ref{sec:storage_engine}), reducing ingestion peak memory and converting the freed DRAM into higher serving-cache efficiency, while preserving full OpenSearch compatibility and support rich predicates, filtering, and ranking options (\S\ref{sec:query_engine}).

\section{System Overview}
\label{sec:overview}
\begin{figure}[t]
  \centering
  \includegraphics[width=0.49\textwidth]{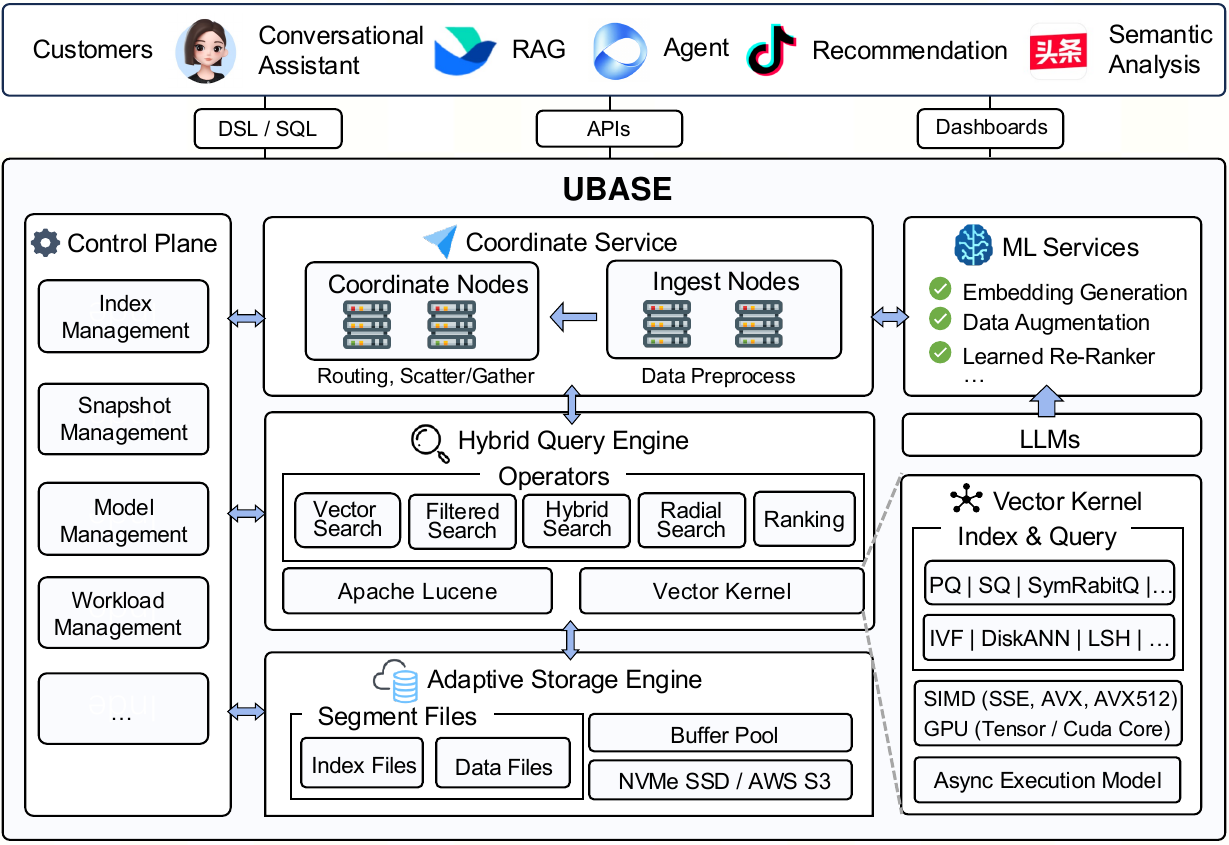}
  \caption{System architecture of \ByteX{}}
  \label{fig:overview}
\end{figure}
Fig.~\ref{fig:overview} shows the \ByteX{} architecture. Externally, it preserves full OpenSearch compatibility: clients, REST APIs, SDKs, ingestion pipelines, dashboards, and index lifecycle tools work unchanged. This lets existing deployments adopt \ByteX{} without rewriting workflows.

Internally, \ByteX{} treats vectors as first-class citizens. Instead of treating vectors as add-ons, it co-designs vector representation, storage layout, and hybrid query execution to boost search efficiency while keeping OpenSearch’s robustness for non-vector workloads.

The engine redesign targets three core components for production-scale vectors:

\begin{itemize}[leftmargin=*, nosep]
\item \textbf{Quantization-aware vector kernel} (\S\ref{sec:vector_kernel}): \textsf{SymRaBitQ} enables distance estimation, index construction, and segment merge without repeated full-precision materialization.
\item \textbf{Adaptive storage engine} (\S\ref{sec:storage_engine}): controls segment lifecycle, memory/SSD placement, caching, and async I/O, addressing the large memory-to-SSD cost gap in production.
\item \textbf{Hybrid query engine} (\S\ref{sec:query_engine}): unifies vector similarity, lexical relevance, structured predicates, score fusion, and reranking in an OpenSearch-compatible distributed plan.
\end{itemize}

\para{Ingestion \& Query Workflow.}
\ByteX{} separates write freshness from serving efficiency. Newly refreshed data first becomes searchable as lightweight fresh segments, while background merge bounds query fanout and progressively upgrades segments from raw vectors, to IVF indexes, to graph indexes. For graph-indexed segments, \ByteX{} builds and merges directly in the \textsf{SymRaBitQ} quantized space, avoiding repeated full-precision vector materialization. The segment abstraction also decouples indexing from placement: the same index can serve from DRAM, hybrid memory--SSD, or SSD-resident storage without reindexing. 

On the read path, \ByteX{} preserves native DSL semantics but executes hybrid retrieval as one distributed plan. BM25, vector search, predicates, range filters, score fusion, and reranking are optimized together rather than stitched across separate services. The query engine is placement-oblivious: DRAM segments use in-memory traversal, while SSD-resident graph records use a record-level buffer pool and asynchronous traversal to overlap I/O with distance computation. Appendices~\ref{sec:appendix_write}--\ref{sec:appendix_dsl} detail the ingestion path, the serving path, and a representative native DSL query.

\section{Quantization-Aware Vector Kernel}
\label{sec:vector_kernel}
Existing SOTA quantization schemes, such as RabitQ/BBQ \cite{DBLP:journals/pacmmod/GaoL24, bbq}, are asymmetric: they quantize data vectors but evaluate them against full-precision query vectors. Thus, they do not directly support quantized-domain data-to-data distance evaluation, which \ByteX needs during index construction. Falling back to full precision would require an extra uncompressed data copy, which amplifies memory usage. \ByteX therefore introduces \textsf{SymRaBitQ}, a symmetric scheme that unifies data-to-query and data-to-data evaluation in the same quantized-space primitive.



\subsection{SymRaBitQ Quantization}
\label{sec:quant}

\textsf{SymRaBitQ} is the distance primitive used by the \ByteX vector kernel.  Throughout this section, raw vectors are denoted by $\bm{x}_r,\bm{q}_r \in \mathbb{R}^D$.  When a vector is encoded relative to a centroid $\bm{c}$, \eg in IVF or partitioned graph indexes, \ByteX uses normalized residuals
\begin{equation}\small
\bm{o}:=\frac{\bm{x}_r-\bm{c}}{\|\bm{x}_r-\bm{c}\|},
\qquad
\bm{q}:=\frac{\bm{q}_r-\bm{c}}{\|\bm{q}_r-\bm{c}\|}.
\end{equation}
We write $\bar{\bm{o}}$ and $\bar{\bm{q}}$ for their quantized representations, and store for every encoded data vector the residual norm $\rho_{\bm{o}}=\|\bm{x}_r-\bm{c}\|$ and the self-correlation $\gamma_{\bm{o}}=\langle \bar{\bm{o}},\bm{o}\rangle$.
Algorithm~\ref{alg:symrabitq-quant} summarizes the encoding path used for both data and query residuals.  The variable $\bm{x}$ is instantiated as $\bm{o}$ for stored vectors and $\bm{q}$ for queries.

\begin{algorithm}[t]
\small
\caption{\textsf{SymRaBitQ} Quantization}
\label{alg:symrabitq-quant}
\LinesNumbered
\KwIn{Vector $\bm{x}_r$, centroid $\bm{c}$, bit width $B$, rotation matrix $P$}
\KwOut{Encoded record $z_{\bm{x}}=(\bar{\bm{x}},\bar{\bm{x}}_0,\gamma_{\bm{x}},\rho_{\bm{x}})$}

$\rho_{\bm{x}}\leftarrow\|\bm{x}_r-\bm{c}\|$, $\bm{x}\leftarrow(\bm{x}_r-\bm{c})/\rho_{\bm{x}}$\;
$\bm{y}\leftarrow P^{-1}\bm{x}$\;
$\bm{u}\leftarrow$ nearest $B$-bit integer-lattice code to $\bm{y}$\;
$\tilde{\bm{u}}\leftarrow\bm{u}-(2^B-1)\bm{1}_D/2$, $\bar{\bm{x}}\leftarrow P\tilde{\bm{u}}/\|\tilde{\bm{u}}\|$\;
$\bar{\bm{x}}_0\leftarrow$ first-bit/sign projection of $\bm{u}$\;
$\gamma_{\bm{x}}\leftarrow\langle\bar{\bm{x}},\bm{x}\rangle$\;
\Return{$z_{\bm{x}}=(\bar{\bm{x}},\bar{\bm{x}}_0,\gamma_{\bm{x}},\rho_{\bm{x}})$}
\end{algorithm}

\para{Symmetric estimator.}
Given two encoded residual unit vectors, \textsf{SymRaBitQ} estimates their inner product as

\begin{equation}\small
\label{eq:symrabitq_est}
\widehat{\langle \bm{o}, \bm{q} \rangle}
\;=\;
\frac{\langle \bar{\bm{o}}, \bar{\bm{q}} \rangle}
{\langle \bar{\bm{o}}, \bm{o} \rangle\,\langle \bar{\bm{q}}, \bm{q} \rangle}.
\end{equation}
For data-to-data comparisons, both denominator terms are precomputed; for query-to-data comparisons, only $\langle \bar{\bm{q}},\bm{q}\rangle$ is computed online once. The corresponding squared $L_2$ distance used by the graph and IVF kernels is
\begin{equation}\small
\label{eq:symdist}
\operatorname{SymDist}_{B}(\bar{\bm{o}},\bar{\bm{q}};\bm{c})
= \rho_{\bm{o}}^2 + \rho_{\bm{q}}^2
-2\rho_{\bm{o}}\rho_{\bm{q}}\cdot \widehat{\langle \bm{o},\bm{q}\rangle},
\end{equation}
where $B$ denotes the bit width of the code.  Inner-product search is handled by the same estimator after applying the standard residual reconstruction terms.

\para{Code-only computation.}
The numerator $\langle\bar{\bm{o}},\bar{\bm{q}}\rangle$ in Eq.~(\ref{eq:symrabitq_est}) is computed directly from the integer codes, not by materializing rotated floating-point vectors.  
For $B=1$, the computation reduces to sign agreement and can be evaluated with bit operations and popcount.  For the $B=5$ construction code used in \ByteX, it becomes a compact integer-code dot product; the graph and IVF kernels below call Eq.~(\ref{eq:symdist}) directly rather than a generic full-precision distance.

\begin{remark}[Relation to RaBitQ]
\textsf{SymRaBitQ} retains the code construction of
RaBitQ/ExtRaBitQ~\cite{DBLP:journals/pacmmod/GaoL24,
gao2025practical, rabitq_library} but changes the distance estimator.
To recap, RaBitQ uses
\begin{equation}\small
\label{eq:rabitq_est}
\widehat{\langle\bm{o},\bm{q}\rangle}_{\mathrm{RaBitQ}}
=
\frac{\langle\bar{\bm{o}},\bm{q}\rangle}
{\langle\bar{\bm{o}},\bm{o}\rangle}.
\end{equation}
This estimator is asymmetric: the data residual is represented
by its quantized form $\bar{\bm{o}}$, while the query residual
$\bm{q}$ remains in full precision.
In contrast, \textsf{SymRaBitQ} quantizes both operands,
replacing the code-to-raw inner product
$\langle\bar{\bm{o}},\bm{q}\rangle$ with the code-to-code
inner product $\langle\bar{\bm{o}},\bar{\bm{q}}\rangle$
and introducing the corresponding normalization factor
$\langle\bar{\bm{q}},\bm{q}\rangle$ in the denominator,
as shown in Eq.~(\ref{eq:symrabitq_est}).
Further details on RaBitQ and ExtRaBitQ are provided in
Appendix~\ref{sec:appendix_rabitq}.
\end{remark}

\begin{remark}[Benefits of symmetry]
RaBitQ is sufficient for search, but can become a bottleneck
during index construction, which requires many pairwise
distance computations among stored vectors.
RaBitQ must retain or fetch full-precision operands for
these comparisons, whereas \textsf{SymRaBitQ} uses only
compact codes and precomputed scalars.
This extends the benefits of quantization from query
processing to index construction and maintenance (\S \ref{sec:graph_based_kernel} and \S \ref{sec:ivf_based_kernel}). 
\end{remark}

\para{Accuracy guarantee.}
\textsf{SymRaBitQ} retains the key properties for an ANN kernel: it has high estimation accuracy and stays close to its expected value with high probability.

\begin{theorem}[SymRaBitQ Estimator]
\label{th:estimator}
The mean $\mathbb{E}[\widehat{\langle\bm{o},\bm{q}\rangle}]$
depends solely and monotonically on $\langle\bm{o},\bm{q}\rangle$, enabling accurate estimation of $\langle\bm{o},\bm{q}\rangle$ based on $\mathbb{E}[\widehat{\langle\bm{o},\bm{q}\rangle}]$. Moreover,
for one-bit \textsf{SymRaBitQ}, a sufficiently large $D$, and $\epsilon_0>0$,
\begin{equation}
\small
\label{eq:symrabitq_centered_bound}
\mathbb{P}\!\left\{
\left|
\widehat{\langle\bm{o},\bm{q}\rangle}
-\mathbb{E}[\widehat{\langle\bm{o},\bm{q}\rangle}]
\right|>
\frac{6\epsilon_0+4}{\gamma_{\bm{o}}\gamma_{\bm{q}}\sqrt{D-2}}
\right\}
\le 8e^{-\epsilon_0^2/2}.
\end{equation}
\end{theorem}

The error bound can be concisely presented as
\begin{equation}
\small
\label{eq:symrabitq_concise_bound}
\begin{gathered}
|\widehat{\langle\bm{o},\bm{q}\rangle}
 -\mathbb{E}[\widehat{\langle\bm{o},\bm{q}\rangle}]|
 =O\!\left(\frac{1}{\sqrt D}\right)
 \text{with high probability}.
\end{gathered}
\end{equation}

\begin{proof}[Proof]
Full proofs are in Appendix~\ref{sec:theorem_4.1_proof}.
\end{proof}

\begin{figure}[t]
    \centering
    \includegraphics[width=0.95\linewidth]{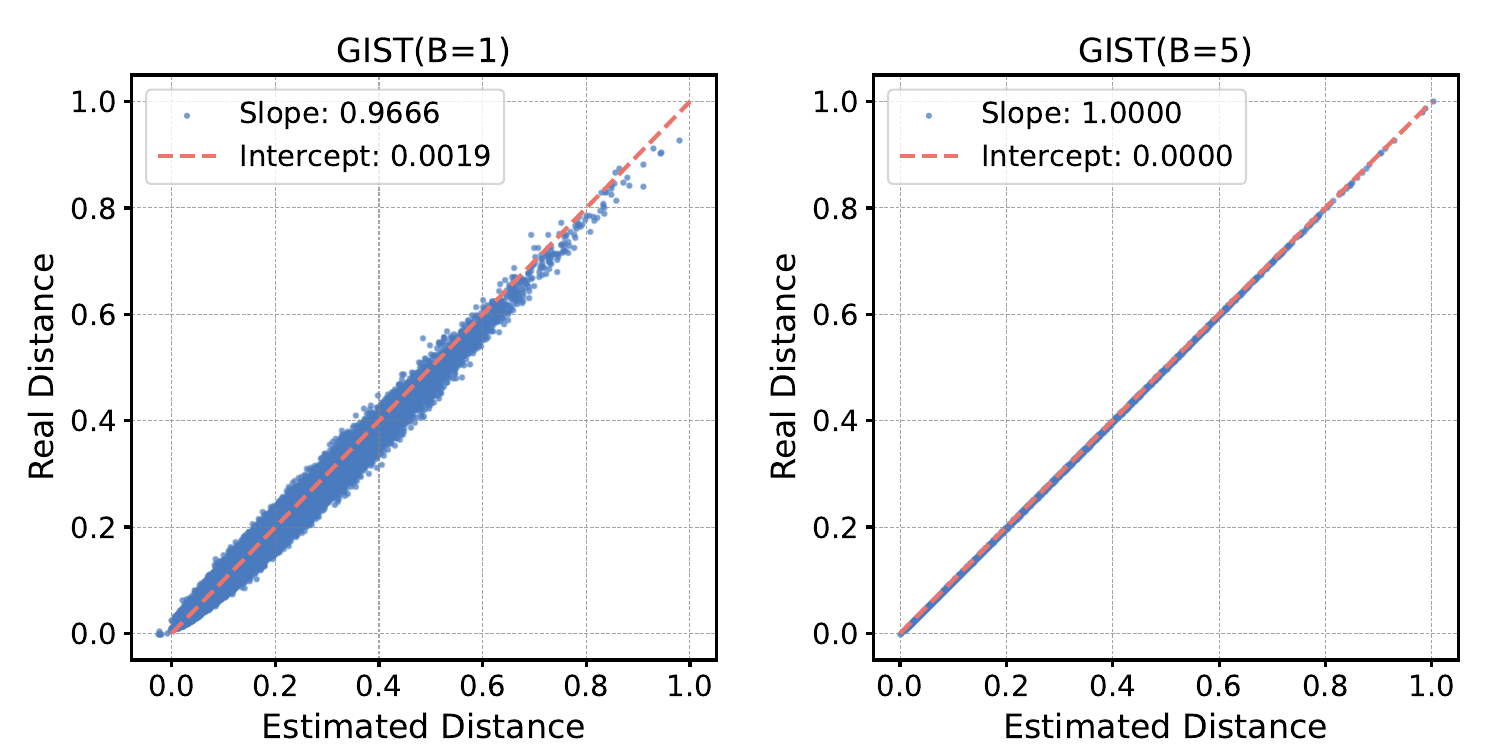}
    \caption{Comparison between estimated distance and real distance on GIST under $B=1$ and $B=5$}
    \label{fig:linear_comparison}
\end{figure}

\begin{example}
  We further provide an empirical example in Fig.~\ref{fig:linear_comparison} to show the estimator’s practical accuracy. The samples lie close to $y=x$, with fitted slopes near $1$ and intercepts near $0$, indicating that the estimated and true distances match closely.
\end{example}

\subsection{Graph-based Kernel}
\label{sec:graph_based_kernel}

Large merged segments in \ByteX use a DiskANN/Vamana-style graph because the construction cost is amortized across many future queries (\S\ref{sec:preliminary}); the graph background and full pseudocode are given in Appendix~\ref{sec:appendix_graph_background}.  Our contribution is not a new graph topology.  The graph kernel keeps DiskANN's candidate discovery, robust pruning, and backward-edge repair, but changes the representation and every distance oracle used by these steps.

\para{Quantization-aware Construction.}
Before graph construction starts, each assigned vector is encoded as a \textsf{SymRaBitQ} record $z_{\bm{o}}=(\bar{\bm{o}},\bar{\bm{o}}_0,\gamma_{\bm{o}},\rho_{\bm{o}})$, where $\gamma_{\bm{o}}=\langle\bar{\bm{o}},\bm{o}\rangle$ and $\rho_{\bm{o}}=\|\bm{x}_r-\bm{c}\|$ are precomputed.  After this point, the graph builder no longer calls a full-precision distance function.  In \textsf{FindCandidates} (Algorithm~\ref{alg:greedy}), the greedy frontier is ranked by $\operatorname{SymDist}_{B}$ between encoded records.  In \textsf{RobustPrune} (Algorithm~\ref{alg:robustprune}), both target-to-candidate distances and candidate-to-candidate diversity checks also use $\operatorname{SymDist}_{B}$. Backward-edge repair uses the same pruning routine when a neighbor exceeds the degree bound.  
This change is precisely where \textsf{SymRaBitQ} differs from a query-only quantizer. Standard RaBitQ can score a quantized data vector against a raw query, whereas graph construction requires distances between stored vectors. For a point $\bm{o}$, \textsf{RobustPrune} repeatedly selects the closest remaining candidate $\bm{p}$ as an out-neighbor and removes candidates $\bm{v}$ satisfying $\alpha \cdot d(\bm{p},\bm{v}) \le d(\bm{o},\bm{v})$, where $\alpha \ge 1$ controls pruning. Both distances involve pairs of stored vectors.  Without a symmetric estimator, these comparisons would fetch raw candidate vectors from disk or keep them in memory.  With \textsf{SymRaBitQ}, all such comparisons read only $B$-bit codes and per-vector scalars, so graph construction operates in the quantized space. Algorithm~\ref{alg:index} shows the complete construction flow.
 
\para{Quantization-aware Search.}
The query path uses the same record format. A query is encoded once for the searched partition, producing $z_{\bm{q}}=(\bar{\bm{q}},\bar{\bm{q}}_0,\gamma_{\bm{q}},\rho_{\bm{q}})$.  During greedy traversal (Algorithm~\ref{alg:graph-search}), \ByteX uses $\operatorname{SymDist}_{1}$ on the first-bit projections $\bar{\bm{o}}_0$ and $\bar{\bm{q}}_0$ to expand the graph with a bit operation plus SIMD popcount (\texttt{vpopcntq}). Only the small promising set retained by traversal is rescored with $\operatorname{SymDist}_{B}$ using the full code.  Thus construction uses the higher-precision $B$-bit estimator to preserve graph quality, while search uses a coarse-to-fine version of the same estimator to reduce memory bandwidth and CPU work.

\vspace{0.1em}
\para{Impact.}
The immediate effect is that raw vectors are removed from candidate search, pruning, backward-edge repair, and most graph-search distance computations.  The memory footprint of construction becomes proportional to compact codes plus graph metadata rather than high-dimensional floating-point vectors; random I/O to raw vectors is eliminated from the recurring merge path.  In our implementation, $B=5$ preserves more than 98\% distance-estimation accuracy and causes less than 0.01 recall loss under the same search parameters.
On GIST-1M (Table~\ref{tab:indextime} in \S \ref{sec:exp}), \textsf{SymRaBitQ} reduces peak build memory by  80\% and total build time by 60\%, demonstrating the significant benefits of quantization for index construction.

\vspace{-0.5em}
\subsection{IVF-based Kernel}
\label{sec:ivf_based_kernel}

For medium-size segments (\S\ref{sec:preliminary}), \ByteX builds an IVF index because it is cheap to construct and the segment may still participate in later merges.  IVF background is given in Appendix~\ref{sec:appendix_ivf_background}; the key point in the vector kernel is that IVF list construction and list scanning also use \textsf{SymRaBitQ} as the distance primitive.

\para{Index primitives.}
\ByteX first quantizes all vectors with \textsf{SymRaBitQ} using the segment centroid. The entire $k$-means clustering procedure then operates in the quantized space via $\operatorname{SymDist}_{B}$. The quantized codes $(\bar{\bm{o}},\bar{\bm{o}}_0,\gamma_{\bm{o}},\rho_{\bm{o}})$  are stored in their assigned inverted lists, while full-precision vectors remain on disk and are no longer needed for list construction or later segment merge.  This keeps medium-segment indexing lightweight while preserving the same compressed representation used by graph segments.


\para{Search primitives.}
At query time, \ByteX encodes the query with \textsf{SymRaBitQ}
and probes the closest $n_{\mathit{probe}}$ centroids in the IVF index.  For each probed centroid, 
we evaluate $\operatorname{SymDist}_{1}$ on 1-bit codes to obtain coarse distances using \texttt{vpopcntq}. Candidates are screened using these coarse estimates before accessing the remaining bits. The survivors are refined with $\operatorname{SymDist}_{B}$ using FastScan-style lookup tables~\cite{DBLP:journals/pvldb/AndreKS15,DBLP:conf/mir/AndreKS17,DBLP:journals/pami/AndreKS21}, and the top-$k$ heap is updated. In this path, \textsf{SymRaBitQ} improves both bandwidth and CPU efficiency: most list entries are judged by compact 1-bit codes, and only promising entries pay for the full $B$-bit estimate.

\section{Vector Storage Engine}

\ByteX's storage engine manages the lifecycle of vector segments beneath the OpenSearch-compatible interface. At a high level, it makes newly written segments searchable at refresh time, bounds query fanout through background merges, upgrades per-segment indexes from scan to IVF and then to graph-based search, and places each segment on memory-only, hybrid memory--disk, or disk-resident serving paths. We inherit OpenSearch's data layout; the ingestion path is described in Appendix~\ref{sec:appendix_write}.
   
The rest of this section focuses on the optimizing of serving disk-resident segments inspired by VeloANN \cite{zhao2026veloann}. Sec.~\ref{sec:cache_management} introduces a fine-grained cache mechanism that largely reduces disk I/Os. Sec.~\ref{sec:async_exec_model} describes the asynchronous execution model that reduces tail latency by overlapping graph traversal with SSD reads. 


\label{sec:storage_engine}
\label{sec:ssd_vector_serving}
\begin{figure}[t]
  \centering
  \includegraphics[width=0.47\textwidth]{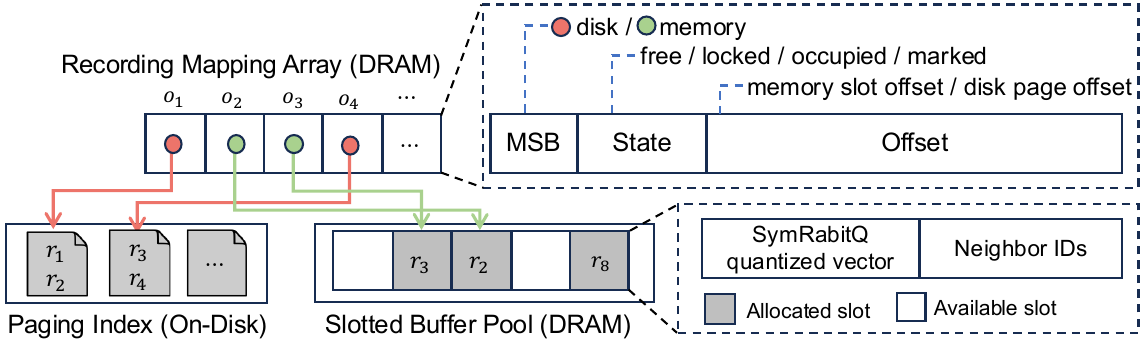}
  \caption{Record-level cache layout}
  \label{fig:cache_layout}
\end{figure}

\subsection{Fine-Grained Record Cache}
\label{sec:cache_management}
\ByteX follows OpenSearch's sharded execution model, with each data node independently managing a local cache for the shards it hosts. To use each node's limited DRAM efficiently, the cache must prioritize data that actually advances graph traversal. Traditional page-granular caching is too coarse.
\ByteX{} therefore employs fine-grained, record-level caching.  

Specifically, each record stores the compact vector payload for distance estimation and a bounded neighbor list for traversal.  
As shown in Fig.~\ref{fig:cache_layout}, a dense record-mapping array indexed by vertex ID gives $O(1)$ access to each record through pointer-swizzling-style logical pointers~\cite{leis2024leanstore}.  Each entry of the record-mapping array is a 64-bit logical pointer consisting of a residency bit, a 2-bit state field, and a 61-bit offset.  An MSB of 0 indicates that the record is on disk and the offset is the disk-page offset; an MSB of 1 indicates that the record is cached in memory and the offset is the slot offset in the slotted buffer region.  The state field records transient cache states, which \ByteX uses to coordinate record loading, prefetching, and eviction under concurrent queries.
To manage these operations, \ByteX uses a lightweight lock-free protocol.  On a miss, the requester locks the disk-resident mapping entry, obtains a free slot or runs clock-style eviction, issues a non-blocking SSD read, copies the fetched record into the slot, and publishes the memory-slot offset as an occupied resident record.  Concurrent demand loads or prefetches that observe the disk entry in the locked state wait for the same load to complete rather than issuing duplicate reads.  During eviction, the clock hand first demotes occupied resident entries to marked; a later access promotes the marked entry back to occupied, while an unaccessed marked entry can be locked, its slot reclaimed and returned to the free list, and its mapping entry restored to the disk location.

\vspace{-0.5em}
\subsection{Asynchronous ANN Execution}
\label{sec:async_exec_model}
Fine-grained caching improves DRAM utilization, but cache misses can still occur during graph traversal.  \ByteX treats those misses as scheduling points inside the SSD-resident vector operator.  Each ANN search runs as a coroutine that carries its candidate frontier, visited vertices, and current top-$k$ results.  The traversal is cache-aware at the candidate-beam level: it expands resident candidates first when they are available, while issuing asynchronous reads or prefetches for promising on-disk candidates.  If no ready candidate can make progress, the coroutine yields after submitting the required SSD read.  The worker then advances another ready ANN coroutine and resumes the suspended search when the read completes.  This design complements the record cache: the cache reduces avoidable SSD traffic, while coroutine scheduling overlaps the remaining SSD latency with useful vector-search computation.

\section{Hybrid Query Engine}
\label{sec:query_engine}
\label{sec:complex_query_plan}
\label{sec:beyond_ann_search}

Production workloads (\S \ref{sec:preliminary}) are rarely plain top-$k$ ANN retrieval: they often combine vector similarity with lexical matching, structured predicates, ranking, reranking, and threshold-based retrieval. To support diverse workloads efficiently, \ByteX integrates vector search into an OpenSearch-compatible distributed query engine. It retains OpenSearch's DSL, coordinator--shard execution model, and score fusion pipeline, while redesigning shard-local vector execution by combining \textsf{SymRabitQ}, the storage runtime described earlier, and three kinds of specific hybrid-retrieval mechanisms presented below.
We refer the reader to Appendices.~\ref{sec:appendix_read} and ~\ref{sec:appendix_dsl} for query path details and a hybrid query DSL example.

\vspace{0.2em}
\para{Vector + Full-text Search.} Production workloads often benefit from combining lexical matching, which excels at exact term retrieval, \eg searching logs by error codes, with vector search that captures semantic relationships beyond literal expressions.\ByteX provides flexible hybrid retrieval within a unified query framework, where query terms are processed through Lucene's inverted index with BM25 scoring and query embeddings are evaluated by the vector search engine. Users can choose from multiple score normalization strategies, including \texttt{min-max}, \texttt{L2}, \texttt{z-score}, and reciprocal rank fusion (\texttt{RRF}), as well as several score fusion methods, such as \texttt{weighted arithmetic mean}, \texttt{geometric mean}, and \texttt{harmonic mean}. These configurable options allow fine-grained control over the balance between lexical and semantic signals, enabling retrieval behavior to be tailored to different applications and relevance objectives.




\para{Radial Search.}
Radial search is useful when applications specify similarity or distance constraints rather than a fixed top-$k$. Recommendation workloads, \eg Customer C, often require retrieving results within a relevance threshold (\eg \texttt{max\_distance} or \texttt{min\_score}) while excluding near-duplicates via an exclusion threshold (\eg \texttt{min\_distance} or \texttt{max\_score}). To efficiently support these query patterns, \ByteX designs a \textit{threshold-guided neighborhood expansion} strategy on graph-based indexes. The search starts from an entry point obtained via a top-1 ANN lookup and expands neighbors in a breadth-first manner, interleaving threshold evaluation with traversal. Nodes satisfying the bound are emitted and recursively explored, while others are pruned early. The 
search terminates when the frontier is exhausted or when a system-defined exploration limit is reached, ensuring bounded work while preserving recall within the specified range.

\para{Filtered Vector Search.}
Production workloads often include structured predicates such as tenants, categories, tags, timestamp ranges, and access-control policies. \ByteX supports three filtering modes: pre-filtering, post-filtering, and in-filtering. Pre-filtering applies predicates before vector search and is effective for highly selective conditions, while post-filtering evaluates predicates after ANN retrieval and is suitable for low-selectivity predicates. When neither strategy is effective, \ByteX enables an in-filter execution mode. The predicate is first materialized as a bitmap over valid vector IDs, and ANN search retrieves a small set of high-similarity seeds. The vector operator then expands from these seeds while checking the bitmap on the fly, admitting only qualifying candidates. Expansion continues until sufficient valid results are collected.

\vspace{-1em}
\section{Experiment}
\label{sec:exp}
We start by evaluating \ByteX using public benchmarks and production-derived workloads on a controlled testbed, and then provide further quantitative results showing that \ByteX works well at the trillion scale.
Specifically, our evaluation answers seven production-facing questions:
\begin{itemize}[leftmargin=18pt,itemsep=1pt,topsep=2pt]
    \item[\textbf{Q1}:] Can quantization-aware vector kernel in \ByteX reduce recurring indexing cost, improve the QPS--recall frontier, and lower tail latency? (\S\ref{sec:Vector Search Performance})
    \item[\textbf{Q2}:] Can a unified engine efficiently support filtered, hybrid, and radial retrieval? (\S \ref{sec: Complex Query Performance})
    \item[\textbf{Q3}:] Does \ByteX improve the operating-cost envelope? (\S \ref{sec: Operating Costs})
    \item[\textbf{Q4}:] Can the storage runtime mask SSD stalls under constrained memory and smoothly transition between disk-resident and memory-resident serving? (\S\ref{sec:storage_runtime_analysis})
    \item[\textbf{Q5}:] Does \ByteX scale with data size and cluster size? (\S \ref{sec:scalability_elasticity})
    \item[\textbf{Q6}:] Can \ByteX sustain stable read--write performance during prolonged streaming ingestion? (\S \ref{sec:streaming_microbenchmark})
    \item[\textbf{Q7}:]  Are the techniques still effective at trillion scale? (\S \ref{sec:production_scale_validation})
\end{itemize}

\begin{table}[t]
\small
\centering
\vspace{0.5em}
\caption{Build Time and Peak Memory.}
\label{tab:indextime}
\renewcommand{\arraystretch}{1.3}
\setlength{\tabcolsep}{4pt}
\resizebox{1\linewidth}{!}{
\begin{tabular}{l | l | c | c | c | c}
\specialrule{0.1em}{0pt}{2pt}
Variant
& Dataset
& Indexing Time (s)
& Reduction
& Peak Mem. (GB)
& Reduction \\
\midrule

\multirow{2}{*}{HNSW}
& \cellA Cohere 1M
& \cellA {359.8}
& \cellA {--}
& \cellA {4.8}
& \cellA {--} \\
\cmidrule(lr){2-6}

& \cellB Cohere 100M
& \cellB {41946.2}
& \cellB {--}
& \cellB {476.19}
& \cellB {--} \\

\midrule

\multirow{2}{*}{DiskANN}
& \cellA GIST 1M
& \cellA {587}
& \cellA {--}
& \cellA {3.75}
& \cellA {--} \\
\cmidrule(lr){2-6}

& \cellB Cohere 100M
& \cellB {23826}
& \cellB {--}
& \cellB {453}
& \cellB {--} \\

\midrule

\multirow{2}{*}{\ByteX}
& \cellA GIST 1M
& \cellA \textbf{194+43}
& \cellA \textbf{$\downarrow$ 60\%}
& \cellA \textbf{0.75}
& \cellA \textbf{$\downarrow$ 80\%} \\
\cmidrule(lr){2-6}

& \cellB Cohere 100M
& \cellB \textbf{11033+1463}
& \cellB \textbf{$\downarrow$ 48\%}
& \cellB \textbf{90}
& \cellB \textbf{$\downarrow$ 80\%} \\

\specialrule{0.1em}{0pt}{2pt}
\end{tabular}
}
\end{table}

\subsection{Experiment Setup}
\para{Datasets.} We use four production-derived datasets and six public datasets.  The production-derived datasets include Customer~E's 1.07\,T-vector collection, and samples of 100M, 120M, and 22.36B vectors from Customer~D's 1.17\,T-vector deployment, as described in \S\ref{sec:preliminary}.  Together with Wiki (1M, 768D, cosine) \cite{wiki}, GIST (1M, 960D, L2) \cite{gist}, and Cohere (113M, 1024D, L2) \cite{cohere}, these datasets are used to evaluate throughput, latency, recall, and index construction cost. For end-to-end
hybrid retrieval quality, we use NFCorpus, FIQA, and Quora from a popular benchmark BEIR \cite{DBLP:conf/nips/Thakur0RSG21, DBLP:conf/sigir/KamallooTLMYL24}.

\para{Competitors.} We compare our system with three well-known specialized and generalized vector databases: 
Elasticsearch 8.18\footnote{https://github.com/elastic/elasticsearch}, 
Milvus 2.5\footnote{https://github.com/milvus-io/milvus}, and 
PostgreSQL 17\footnote{http://www.postgresql.org/} with the VectorChord 1.0 extension\footnote{https://github.com/tensorchord/VectorChord/}.
Elasticsearch is a distributed search and analytics engine that has recently incorporated vector search via an HNSW-based index with BBQ \cite{bbq}.
Milvus is a widely used specialized vector database, providing high-performance ANN search with multiple indexing algorithms (\eg IVF, HNSW, SQ). 
PostgreSQL with VectorChord represents the trend of integrating vector search into relational databases, enabling ANN search within SQL-based hybrid workloads.

\para{Environment and Configurations.}
Our controlled benchmarks run on machines equipped with Intel(R) Xeon(R) Platinum 8457C CPUs. \ByteX is deployed as an OpenSearch-compatible cluster with three data nodes and two coordinator nodes. For smaller-scale datasets (Wiki and GIST), each data node uses 8 CPU cores, 64 GB of memory, and 1 TB of network-attached storage, and each coordinator node uses 8 CPU cores and 16 GB of memory. For large-scale datasets (Cohere and In-house), each data node uses 24 CPU cores, 192 GB of memory, and 2 TB of network-attached storage, and each coordinator node uses 16 CPU cores and 32 GB of memory. Storage is connected via 10 Gigabit Ethernet. Baseline systems are provisioned under comparable resource budgets and tuned according to their recommended practices. Unless otherwise stated, all results use the default per-system resource and index settings reported in Appendix~\ref{sec:machine_settings}.

\para{Evaluation Metrics.}
We report Queries Per Second (QPS), latency, and Recall@$k$ for query performance, and normalized discounted cumulative gain (NDCG) for end-to-end ranking quality. Recall is defined as
\begin{equation}
\small
Recall@k = \frac{|R \cap R^*|}{k},    
\end{equation}
where $R$ denotes the returned top-$k$ results and $R^*$ denotes the exact top-$k$ results for a query. NDCG is defined as
\begin{equation}
\small
   \mathit{NDCG@k} = \frac{\mathit{DCG@k}}{\mathit{IDCG@k}},
\quad
\mathit{DCG@k} = \mathrm{rel}_1 + \sum_{i=2}^{k} \frac{\mathrm{rel}_i}{\log_2(i)},
\end{equation}
where $\mathrm{rel}_i$ denotes the relevance score of the result at rank $i$, and $\mathit{IDCG@k}$ is the maximum possible $DCG@k$ under an ideal ranking.
We measure indexing performance by construction time, peak memory consumption, and index size. We also report query cost per million queries and estimated in-house storage cost for cost-effectiveness.
\subsection{Indexing \& Vector Search Performance}
\label{sec:Vector Search Performance}

This subsection evaluates whether \ByteX reduces index-construction cost and improves high-recall serving efficiency.








\begin{figure}[t]
  \centering
  \includegraphics[width=\linewidth]{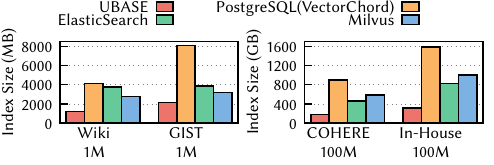}
  \vspace{-1.5em}
  \caption{Index Size}
  \label{fig:index_size}
\end{figure}

 \begin{figure*}[t]
    \centering

    \begin{minipage}[t]{\textwidth}
        \centering
        \includegraphics[width=\linewidth]{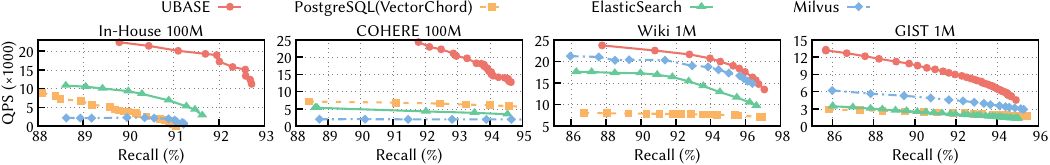}
        \vspace{-3ex}
        \caption{QPS--Recall Trade-off Overview}
        \label{fig:overall_performance}
    \end{minipage}

    \vspace{0.7em}

    \begin{minipage}[t]{\textwidth}
        \centering
        \includegraphics[width=\linewidth]{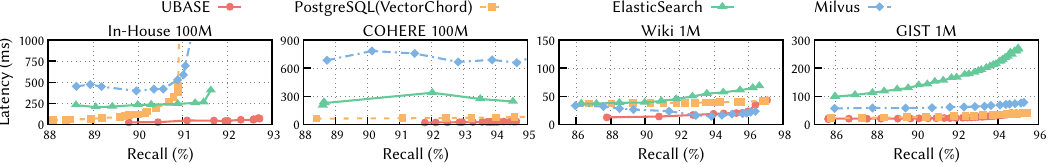}
        \vspace{-3ex}
        \caption{Concurrent P99 Latency--Recall Trade-off.}
        \vspace{-2ex}
        \label{fig:latency-recall-conc-p99}
    \end{minipage}

\end{figure*}

 \begin{figure*}[t]
  \centering
  \includegraphics[width=\linewidth]{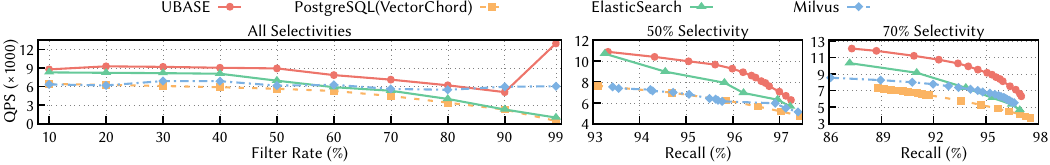}
  \vspace{-3ex}
  \caption{Filtered Search Performance Overview on Wiki 1M}
  \vspace{-2ex}
  \label{fig:filter_performance}
\end{figure*}

\para{Indexing.}
We evaluate the recurring cost of index construction along three dimensions: build time, peak memory, and final index footprint. Table~\ref{tab:indextime} compares \ByteX's quantization-aware builder with DiskANN and HNSW, the graph-construction backends used across our competing systems. Among these, DiskANN serves as our closest controlled baseline for isolating the effect of \texttt{SymRaBitQ}, since \ByteX retains DiskANN's Vamana graph construction algorithm but replaces full-precision distance computations with \texttt{SymRaBitQ} distances. On GIST 1M, \texttt{SymRaBitQ} reduces build time from 587 s to 237 s ($60\%$). Of the 237 s, 43 s is spent on quantization and the remaining 194 s on graph construction in the quantized space. Peak memory is simultaneously reduced from 3.75 GB to 0.75 GB ($80\%$). The same effect holds at production scale: on Cohere 100M, build time drops from 23{,}826 s to 12{,}496 s ($48\%$), while peak memory drops from 453 GB to 90 GB ($80\%$). These results show that \texttt{SymRaBitQ} is not merely a query-time compression technique; it moves quantization into the recurring build and merge path, cutting the memory spikes that contend with online serving during continuous ingestion.

Fig.~\ref{fig:index_size} further shows that \ByteX achieves the smallest end-to-end index footprint across datasets. The reason is architectural: because \ByteX constructs a high-quality graph directly over high-precision quantized vectors, full-precision vectors are no longer needed in the hot serving path and are kept only in remote storage for durability and offline maintenance. On Wiki, \ByteX's index is 1{,}182 MB, only $29\%$ of PostgreSQL and $42\%$ of Milvus, the most space-efficient baseline. The trend persists on the 100M-scale datasets, where retaining full-precision or weakly compressed vectors directly increases serving footprint and storage cost.

The effect is not limited to graph construction. \textsf{SymRabitQ} also reduces the cost of IVF indexing. On GIST 1M, conventional IVF requires 176 s and 3.7 GB peak memory, while \texttt{SymRaBitQ}-enhanced IVF reduces build time to 129 s (27\%) and peak memory to 1.9 GB (49\%). Since IVF shows similar behavior, we omit further IVF results and focus subsequent experiments on graph indexing, which dominates both indexing complexity and serving performance in \ByteX.

\para{Vector Search.}
We next evaluate whether \ByteX improves the throughput--recall and tail-latency--recall frontiers. Since recall above $90\%$ is a common target in our production workloads, we focus on this high-recall regime and compare against Elasticsearch, Milvus, and PostgreSQL/VectorChord.

Fig.~\ref{fig:overall_performance} shows that \ByteX moves the QPS--recall frontier rather than improving only one operating point. On the in-house 100M dataset, \ByteX reaches 21{,}445 QPS at around $90\%$ recall@10, which is 9.5$\times$ higher than Milvus, 5.5$\times$ higher than PostgreSQL, and 2.3$\times$ higher than Elasticsearch. As recall increases, the baselines degrade rapidly around $92\%$ recall, while \ByteX still sustains over 10K QPS. On Cohere, \ByteX reaches 23{,}013 QPS at around $92\%$ recall@10, outperforming Milvus by 11.7$\times$, Elasticsearch by 5.4$\times$, and PostgreSQL by 3.5$\times$.
 
On smaller datasets such as GIST 1M and Wiki 1M, the gap narrows because most systems can keep a larger fraction of the working set in memory and avoid frequent disk I/O. Nevertheless, \ByteX remains on the best QPS--recall frontier. On GIST, it is up to 5$\times$ faster than PostgreSQL and more than 2$\times$ faster than the second-best Milvus configuration. On Wiki, it reaches nearly 25K QPS at $87\%$ recall and still sustains around 15K QPS at $97\%$ recall. Although the second-best baseline changes across datasets, \ByteX consistently provides the best frontier.

\begingroup
\brokenpenalty=100 
The improvement comes from three design choices. First, \ByteX uses a vector-native execution path integrated into the OpenSearch query pipeline, avoiding redundant data movement and cross-engine coordination. Second, \ByteX adopts a quantization-aware vector kernel, where distance computation is performed in the quantized space using fast scan operations such as \texttt{vpopcnt}, reducing per-query computation cost. Third, the quantization scheme preserves distance fidelity well enough to sustain high recall under aggressive pruning and high-QPS settings. 
Fig.~\ref{fig:latency-recall-conc-p99} shows the corresponding concurrent p99 latency results: at comparable recall, \ByteX also maintains lower tail latency than the baselines. Together, the throughput and latency results show that \ByteX improves the high-recall serving frontier on both axes. Additional average, p95, and serial-latency results are reported in Appendix~\ref{sec:appendix_latency}.
\par
\endgroup

\vspace{-1em}
\subsection{Complex Query Performance}
\label{sec: Complex Query Performance}
This subsection evaluates \ByteX’s efficiency across three representative query patterns: filtered search (vector + attribute filters), hybrid search  (vector + lexical matching), and radial search (similarity-threshold retrieval).
\begin{table}[t]
\small
\centering
\setlength{\tabcolsep}{4pt}
\caption{NDCG@10 Results with Different Fusion Strategies.}
\label{tab:ndcg10-only}
\renewcommand{\arraystretch}{1.3}
\resizebox{1\linewidth}{!}{
\begin{tabular}{l | c | c | c | c | c | c}
\specialrule{0.1em}{0pt}{2pt}
Dataset
& BM25
& Dense
& Min-Max
& RRF
& Z-score
& Improv. \\
\midrule

\rowcolor{rowA}
NFCorpus
& 0.3040
& {0.2321}
& {0.3271}
& \textbf{0.3327}
& 0.3288
& \textbf{+43.3\%} \\
\midrule

\rowcolor{rowB}
FIQA
& 0.2386
& {0.1983}
& 0.2890
& 0.3071
& \textbf{0.3281}
& \textbf{+65.5\%} \\
\midrule

\rowcolor{rowC}
Quora
& {0.7412}
& 0.7469
& 0.8189
& 0.8300
& \textbf{0.8666}
& \textbf{+16.9\%} \\
\specialrule{0.1em}{0pt}{2pt}
\end{tabular}
}
\end{table}

\begin{figure}[t]
  \centering
  \includegraphics[width=\linewidth]{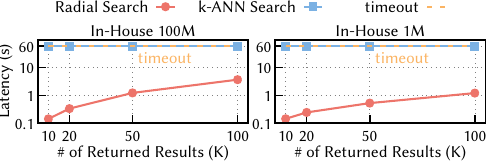}
  \caption{Radial Search Performance Overview}
  \label{fig:range-performance}
\end{figure}

\para{Filtered Search.} Fig.~\ref{fig:filter_performance} reports QPS at approximately 96\% recall@10 as predicate selectivity varies from 10\% to 99\%. We make four observations. First, \ByteX consistently achieves the highest throughput across the entire selectivity spectrum. Second, all systems experience throughput degradation relative to pure vector search. Although scalar filtering itself is inexpensive, its overhead becomes significant when vector search latency is already only a few milliseconds. Third, throughput decreases as predicates become more selective. With fewer candidates satisfying the predicate, ANN search must explore a larger portion of the graph to collect enough valid results. Finally, at the extreme selectivity of 99\%, \ByteX exhibits a noticeable throughput advantage. This is because \ByteX automatically switches to pre-filtering when the estimated candidate set falls below a configurable threshold (10{,}000 in this experiment).
The middle and right panels provide a finer-grained comparison at 50\% and 70\% selectivity. Across all recall levels, \ByteX consistently delivers higher throughput than competing systems, demonstrating the effectiveness of its adaptive filtering strategy.



\para{Hybrid Search.} 
Table~\ref{tab:ndcg10-only} evaluates end-to-end retrieval quality by comparing traditional lexical retrieval (BM25), dense vector retrieval, and three score-fusion strategies supported by \ByteX's query engine in Sec.~\ref{sec:beyond_ann_search}. The results show that lexical and semantic signals are complementary: fusion consistently outperforms BM25 and dense retrieval alone. No single normalization strategy dominates across all datasets. Z-score fusion performs best on FIQA and Quora, improving NDCG@10 by 65.5\% and 16.9\%, respectively, while RRF performs best on NFCorpus with a 43.3\% improvement.

\begin{table}[t]
\centering
\caption{Operating cost (USD) on In-House 100M.}
\label{tab:price}

\small
\renewcommand{\arraystretch}{1.3}
\setlength{\tabcolsep}{4pt}
\setlength{\doublerulesep}{2pt}

\begin{adjustbox}{max width=\linewidth}
\begin{tabular}{c | c | c | c | c}
\toprule
Storage mode
& Method
& \multicolumn{1}{c|}{QPS}
& \multicolumn{1}{c|}{Monthly Cost}
& \multicolumn{1}{c}{Cost / 1K QPS} \\
\midrule

\multirow{4}{*}{Memory-only}
& \cellcolor{rowB}Milvus
& \cellcolor{rowB}2,262
& \cellcolor{rowB}2,790
& \cellcolor{rowB}1,233 \\

\hhline{~|-|-|-|-}
& \cellcolor{rowB}PG
& \cellcolor{rowB}3,929
& \cellcolor{rowB}2,790
& \cellcolor{rowB}710 \\

\hhline{~|-|-|-|-}
& \cellcolor{rowB}ES
& \cellcolor{rowB}9,342
& \cellcolor{rowB}2,790
& \cellcolor{rowB}299 \\

\hhline{~|-|-|-|-}
& \cellcolor{rowA}\ByteX
& \cellcolor{rowA}\textbf{21,445}
& \cellcolor{rowA}2,790
& \cellcolor{rowA}\textbf{130} \\

\hhline{=====}

\multirow{4}{*}{Disk-only}
& \cellcolor{rowB}Milvus
& \cellcolor{rowB}40
& \cellcolor{rowB} 382
& \cellcolor{rowB} 9,550 \\

\hhline{~|-|-|-|-}
& \cellcolor{rowB}PG
& \cellcolor{rowB}13
& \cellcolor{rowB} 382
& \cellcolor{rowB} 29,385 \\

\hhline{~|-|-|-|-}
& \cellcolor{rowB}ES
& \cellcolor{rowB}0.4
& \cellcolor{rowB} 382
& \cellcolor{rowB} 955,000 \\

\hhline{~|-|-|-|-}
& \cellcolor{rowC}\ByteX
& \cellcolor{rowC}\textbf{63}
& \cellcolor{rowC}382
& \cellcolor{rowC}\textbf{6,063} \\

\bottomrule
\end{tabular}
\end{adjustbox}
\end{table}

\para{Radial Search.}  This experiment evaluates the radial search algorithm introduced in Sec.~\ref{sec:beyond_ann_search}. As shown in Fig.~\ref{fig:range-performance}, we vary the radius so that the search region contains approximately $\{10\text{k}, 20\text{k}, 50\text{k}, 100\text{k}\}$ vectors, and evaluate both the 100M in-house dataset and a random 1M subset. A naive implementation that emulates radial search with kNN at $k \in \{10\text{k}, 20\text{k}, 50\text{k}, 100\text{k}\}$ does not finish within the one-minute timeout.  In contrast, \ByteX's threshold-guided neighborhood expansion returns 10K results in 140 ms on both datasets.  Even at 100K returned results, latency remains practical: 3.68s on the 100M dataset and 1.2s on the 1M subset.  These results show that radial search can support large-range queries without reducing them to repeated large-$k$ ANN calls.

\vspace{-1em}
\subsection{Operating Costs}
\label{sec: Operating Costs}
Table~\ref{tab:price} reports the monthly cost of keeping the 100M 2048D dataset online and the serving cost per 1K QPS.  
Under the memory-only setting, all systems operate under the same monthly budget of $\$2790$: Milvus, PG, and ES reach 2262, 3929, and 9342 QPS, respectively, while \ByteX reaches 21445 QPS, with a serving cost of $\$130$ per 1K QPS, 56\%-89\% lower than other memory-only competing systems.  
For disk-only deployment, \ByteX targets minimum standing cost rather than peak throughput. It reduces monthly cost to $\$382$, which is $7.3\times$ cheaper than the memory-only \ByteX, while still serving 63 QPS for low-duty-cycle or archival workloads. Compared with the disk-only baselines, \ByteX reduces the cost per 1K QPS by 36.5\%--99.4\%. Overall, these experiments show that \ByteX is cost-saving in both end of operating modes: maximizes throughput under a fixed budget or minimizes the cost of keeping and serving a large vector collection. See Appendix~\ref{sec:cost_deployment_config} for deployment and pricing details.

\subsection{Storage Runtime Analysis}
\label{sec:storage_runtime_analysis}
This subsection evaluates \ByteX's storage runtime for SSD-resident vector search using controlled ablations. Experiments are conducted on a single server with 48 CPU cores and 371\,GB of memory, with recall fixed at 90\%. We isolate the effects of record-level caching and asynchronous execution, and then vary the buffer-pool ratio to measure the smooth memory--throughput transition.

\para{Record-Level Cache Effectiveness.}
We first evaluate whether the record-level cache described in Sec.~\ref{sec:cache_management} reduces SSD traffic during graph traversal. Fig.~\ref{fig:async} and Fig.~\ref{fig:iosavings} compare DiskANN with \ByteX(-sync), which enables the 20\%-size record cache. \ByteX-sync consistently outperforms DiskANN as the thread count increases, indicating that record-level caching effectively eliminates a large fraction of random SSD accesses.
Fig.~\ref{fig:iosavings} quantifies the resulting I/O savings. \ByteX achieves a 97.7\% cache hit rate on GIST. Although each query visits 105.99 graph records on average, only 51.90 SSD accesses per query are required (I/O misses + prefetches), compared to 80.9 for DiskANN. Similar benefits are observed on the in-house dataset, indicating that the effectiveness of record-level caching persists at larger scales.

\begin{figure}[t]
  \centering
  \includegraphics[width=\linewidth]{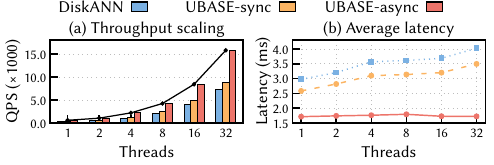}
  \caption{Benefits of Cache and Async }
  \label{fig:async}
\end{figure}
\begin{figure}[t]
    \centering
    \begin{minipage}[t]{1.5in}
        \centering
        \includegraphics[width=\linewidth]{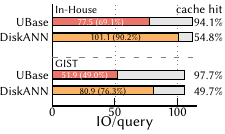}
        \caption{I/O Savings}
        \label{fig:iosavings}
    \end{minipage}%
    \hfill
    \begin{minipage}[t]{1.7in}
        \centering
        \includegraphics[width=\linewidth]{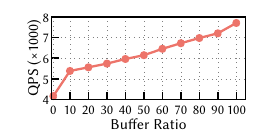}
        \caption{Smooth Transition}
        \label{fig:cache}
    \end{minipage}%
\end{figure}

\para{Asynchronous Execution Effectiveness.}
Fig.~\ref{fig:async} evaluates the impact of asynchronous execution. \ByteX-async issues graph-record fetches asynchronously and switches to other ready query coroutines while I/O is outstanding, whereas \ByteX-sync uses the same cache mechanism but waits synchronously on cache misses. As the number of worker threads increases, \ByteX-async scales substantially better than \ByteX-sync. At 32 threads, it achieves roughly 15K QPS, about 1.8$\times$ higher than \ByteX-sync and 2$\times$ higher than DiskANN. At the same time, \ByteX-async maintains an average latency of approximately 1.7 ms, while \ByteX-sync and DiskANN increase to about 3.5 ms and 4 ms, respectively. These results show that asynchronous execution complements the record cache by turning remaining graph-record I/O stalls into useful inter-query overlap.
 
\para{Smooth Memory--Throughput Transition.}
We further evaluate whether buffer-pool sizing provides a continuous performance knob across storage configurations. We vary the buffer-pool ratio from 0\% to 100\% on the in-house dataset, using a fixed thread count of 12.
 As shown in Fig.~\ref{fig:cache}, increasing the buffer ratio from 0\% to 10\% improves throughput from 4{,}163 to 5{,}381 QPS (1.29$\times$), showing that a small memory budget can already absorb many hot-record accesses in SSD-resident workloads. At a 10\% buffer ratio, \ByteX reaches roughly 70\% of the peak throughput at 100\% (7{,}701 QPS). Throughput then improves monotonically as the buffer ratio grows, giving operators a smooth path from capacity-oriented SSD-resident deployments toward memory-resident high-QPS deployments without changing the query interface or rebuilding the index.


\subsection{Operational Scaling Behavior}
\label{sec:scalability_elasticity}
 
\para{Data-Scale Scalability.}
We first examine how \ByteX scales as the indexed data grows. We construct datasets by randomly sampling $\{1\text{M}, 10\text{M}, 20\text{M},  \ldots, 100\text{M}\}$ vectors from the in-house collection while keeping the serving configuration unchanged. For each data scale, index parameters are tuned to meet the same 90\% recall target, and the corresponding QPS is measured. As shown in Fig.~\ref{fig:scalability}, QPS decreases gradually as the dataset grows, reflecting the increased cost of candidate selection and distance evaluation. Notably, the degradation remains smooth and approximately logarithmic: when the dataset size increases by two orders of magnitude, from 1M to 100M vectors, QPS decreases by only 9\%, from 24{,}410 to 22{,}386. This behavior is consistent with the logarithmic scaling characteristics of graph-based similarity search and provides a predictable performance envelope as collections grow.

 
\begin{figure}[t]
    \centering
    \begin{minipage}[t]{1.6in}
        \centering
        \includegraphics[width=\linewidth]{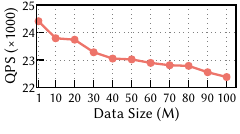}
        \vspace{-1.5em}
        \caption{Scalability}
        \label{fig:scalability}
    \end{minipage}%
    \hfill
    \begin{minipage}[t]{1.6in}
        \centering
        \includegraphics[width=\linewidth]{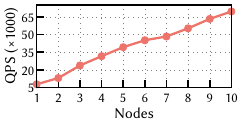}
        \vspace{-1.5em}
        \caption{Elasticity}
        \label{fig:elasticity}
    \end{minipage}%
\end{figure}
\begin{figure}[t]
    \centering
    \begin{minipage}[t]{1.6in}
        \centering
        \includegraphics[width=\linewidth]{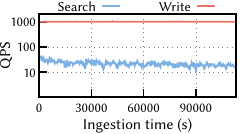}
        \caption{Streaming QPS}
        \label{fig:streamingqps}
    \end{minipage}%
    \hfill
    \begin{minipage}[t]{1.6in}
        \centering
        \includegraphics[width=\linewidth]{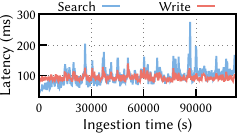}
        \caption{Streaming P99}
        \label{fig:streamingp99latency}
    \end{minipage}%
\end{figure}

\para{Node-Scale Elasticity.}
We further test whether \ByteX can translate additional data nodes into higher serving throughput. UBASE follows OpenSearch's primary--replica architecture. In this experiment, we use a single primary shard on one data node, and each additional data node hosts a replica of that shard.
We vary the number of data nodes from 1 to 10 on the in-house dataset and measure the QPS achieved at the same recall target. As shown in Fig.~\ref{fig:elasticity}, \ByteX achieves near-linear throughput growth as more data nodes are provisioned. 
Queries are distributed across the primary and replicas, allowing the data nodes to serve independent queries concurrently. Each additional data node contributes CPU, memory, and I/O bandwidth while serving a copy of the same index, explaining the near-linear increase in cluster-level throughput. This trend holds under sufficient query concurrency and balanced request distribution, until secondary bottlenecks (\eg coordination overhead or network bandwidth) become dominant.


\subsection{Streaming Micro-benchmark}
\label{sec:streaming_microbenchmark}
This subsection evaluates whether \ByteX can sustain stable serving performance during continuous ingestion. We run a 30-hour micro-benchmark on Cohere: vectors are inserted from an empty index at a fixed rate of 1{,}000 vectors/s, while read queries are issued concurrently under fixed system configurations and index parameters. We continuously monitor query throughput, ingestion throughput, and P99 read/write latencies throughout the run. As shown in Fig.~\ref{fig:streamingqps}, \ByteX sustains the configured ingestion rate throughout the run while keeping query QPS stable, showing that background indexing and segment maintenance do not noticeably disrupt foreground search. Fig.~\ref{fig:streamingp99latency} reports the P99 read and write latencies over the same period. Both remain stable, indicating that continuous writes do not introduce frequent read-path latency spikes or write-side backpressure. We attribute this stability to \ByteX's fast indexing path and controlled merge process. By accelerating segment construction and merge, \ByteX prevents many small streaming segments from accumulating, which would otherwise increase query fan-out and hurt read performance. Overall, these results show that \ByteX provides predictable read--write performance under prolonged streaming workloads.



\subsection{Trillion-Scale Validation}
\label{sec:production_scale_validation}
The performance gains observed on the 100M-vector benchmarks can readily scale to trillion-scale data due to \ByteX's architectural design. \ByteX follows a shard-local design akin to Elasticsearch, where each shard is a bounded, independent unit whose indexing and query processing are invariant to the total collection size. Scaling is achieved by adding shards with proportional resources, not by enlarging individual shards.   Our 100M dataset approximates a complete production shard in deployment, so techniques validated on it—including SymRaBitQ-based construction and the tiered cache policy—apply identically at trillion scale.

To further validate the scale-out argument above, we conduct additional experiments using production data from Customers~D and E. For Customer D, we scale from one complete 120M-vector shard to 22.36B vectors across 180 shards. For Customer E, we evaluate on its 1.07T-vector deployment, as reported in Table~\ref{tab:production_construction} and Table~\ref{tab:production_cache}.

\para{Construction Performance.}
Table~\ref{tab:production_construction} shows that \ByteX maintains nearly identical build time and per-node peak memory when scaling from 120M vectors on one node to 22.36B vectors on 180 nodes. With \textsf{SymRaBitQ}, wall-clock build time increases from 26{,}633\,s to 26{,}743\,s, while per-node peak memory increases from 171.2\,GB to 172\,GB, both by less than $0.5\%$. Total CPU time grows by $176.4\times$ as the node count increases by $180\times$, further supporting approximately stable construction cost at large scale. The full-precision baseline's build time, however, increases from 53{,}007\,s at 120M to 105{,}132\,s at 22.36B. This is because, under the tighter 192 GB/node memory budget, constructing the graph for an entire shard in full precision at once is infeasible (964.8 GB). The baseline therefore adopts subgraph-based construction, yielding a peak memory usage of 190\,GB/node.

\begin{table}[t]
\centering
\vspace{0.5em}
\caption{Production-Scale Index Construction.}
\label{tab:production_construction}
\renewcommand{\arraystretch}{1.6}
\setlength{\tabcolsep}{4pt}
\resizebox{\linewidth}{!}{
\begin{tabular}{c | c | c | c | c | c}
\specialrule{0.1em}{0pt}{2pt}

Dataset
& \makecell{Shards / nodes \\ Resources / node}
& \textsf{SymRaBitQ}
& \makecell{Wall-Clock\\Time (s)}
& \makecell{Total CPU \\Time (s)}
& \makecell{Peak Mem.\\(GB/node)} \\
\midrule

\rowcolor{rowA}
& & On
& \textbf{26{,}633}
& \textbf{26{,}633}
& \textbf{171.2} \\
\hhline{~~|-|-|-|-}

\rowcolor{rowA}
\multirow{-2}{*}{\shortstack{Customer D \\(120M)}}
& \multirow{-2}{*}{
    \shortstack{
    1 / 1 \\
    48C, 1152\,GB
}}
& Off
& 53{,}007
& 53{,}007
& 964.8 \\
\midrule

\rowcolor{rowB}
& & On
& \textbf{26{,}743}
& \textbf{4{,}697{,}073}
& \textbf{172} \\
\hhline{~~|-|-|-|-}

\rowcolor{rowB}
\multirow{-2}{*}{\shortstack{Customer D \\(22.36B)}}
& \multirow{-2}{*}{
    \shortstack{
    180  / 180 \\
    48C, 192\,GB 
}}
& Off
& 105{,}132
& 19{,}199{,}237
& 190 \\

\specialrule{0.1em}{0pt}{2pt}
\end{tabular}
}
\end{table}







\para{Serving Performance.}
Table~\ref{tab:production_cache} reports serving performance under three storage configurations. Customer~D uses the same resources listed in Table~\ref{tab:production_construction}. Pure-memory serving achieves approximately 4K QPS at both 120M and 22.36B vectors. With a $30\%$ cache rate, VeloCache sustains 1K QPS at both scales, with P99 latencies of 15.5\,ms and 317\,ms, respectively. On 120M vectors, disabling VeloCache and using direct I/O reduces throughput to 0.22 QPS ($4{,}545\times$ lower) and increases P99 latency to 22.8\,s ($1{,}471\times$ higher). On 22.36B vectors, throughput drops to 0.13 QPS ($7{,}692\times$ lower), while P99 latency rises to 30.2\,s ($95.3\times$ higher).
Customer~E's 1.07T-vector deployment contains 180 indexes with 72 shards each, hosted on 72 nodes with 24 CPU cores and 192\,GB memory per node.  With a $0.5\%$ cache rate, VeloCache reduces query latency from 9{,}799.058\,s with \texttt{mmap} to 44.22\,s, a $221.6\times$ speedup. These results validate that \ByteX's serving optimizations remain effective in trillion-scale production deployments.

\begin{table}[t]
\small
\centering
\vspace{0.5em}
\caption{Production-Scale Serving Performance.}
\label{tab:production_cache}
\renewcommand{\arraystretch}{1.3}
\setlength{\tabcolsep}{8pt}
\resizebox{\linewidth}{!}{
\begin{tabular}{c | c | c | c}
\specialrule{0.1em}{0pt}{2pt}

Dataset
& Pure Memory
& VeloCache Enabled
& VeloCache Disabled \\
\midrule

\rowcolor{rowA}
\makecell[c]{Customer D\\(120M)}
& \mbox{\makecell[c]{4K QPS\\13.8\,ms P99}}
& \mbox{\makecell[c]{1K QPS, 15.5\,ms P99\\(30\% cache rate)}}
& \mbox{\makecell[c]{0.22 QPS, 22.8\,s P99\\ (direct I/O)}} \\
\midrule

\rowcolor{rowB}
\makecell[c]{Customer D\\(22.36B)}
& \mbox{\makecell[c]{4.01K QPS\\169\,ms P99}}
& \mbox{\makecell[c]{1K QPS, 317\,ms P99\\(30\% cache rate)}}
&  \makecell[c]{0.13 QPS, 30.2\,s P99\\ (direct I/O)} \\
\midrule

\rowcolor{rowC}
\makecell[c]{Customer E\\(1.07T)}
& OOM
& 44.22\,s  (0.5\% cache rate)
& 9{,}799.058\,s (\textsf{mmap}) \\

\specialrule{0.1em}{0pt}{2pt}
\end{tabular}
}
\end{table}

%

\section{Lessons}
\label{sec:lessons}

We distill five lessons from building and operating \ByteX.
\para{Lesson 1: Innovate within compatibility boundary.}
Compatibility shaped \ByteX end-to-end, from API surface to index lifecycle, distributed execution. The lesson is to keep the boundary stable and innovate inside it: \ByteX preserves OpenSearch-facing behavior while redesigning the vector kernel and storage underneath. A standalone vector service might be cleaner, but it would export operational complexity to every application team. Production teams already operate 7,000+ clusters with existing APIs, dashboards, alerts, access controls, and playbooks.

\para{Lesson 2: Quantize the build path, not just the query path.}
Many vector systems apply quantization at query time but still rely on full-precision vectors for index construction.  In a system with sustained ingestion (Customers~D, E), index construction is a recurring production path, and its memory footprint and construction time matter as much as the query-time memory usage and speed, as it will compete the memory and computation resources for serving.  \textsf{SymRaBitQ} enables graph construction and merge entirely in the quantized space.  In our experiments (Table~\ref{tab:indextime}), this reduced peak build memory by 80\% and indexing time by 60\% on GIST 1\,M, with less than 0.01 recall loss under the same search parameters.

\para{Lesson 3: Bridge the cost gap with I/O co-design, not just tiered storage.}
The 60$\times$ memory-to-SSD price gap makes memory-only serving unsustainable for large corpora such as Customer~D's with near trillion-scale vectors.  However, simply placing graph data on SSD hurts tail latency because graph traversal has sparse, pointer-chasing access patterns that break page-level caching.  \ByteX addresses this with record-level buffer pools and asynchronous graph traversal.  In our cache-sensitivity experiment (Fig.~\ref{fig:cache}), a 10\% buffer ratio already achieves $\sim$70\% of peak memory-resident throughput, showing that a small, well-targeted cache can bridge most of the performance gap.

\para{Lesson 4: Make storage placement a tunable segment property.}
Production workloads evolve: what starts as a latency-sensitive service may need to shift toward cost-effective SSD-resident serving as retention or data granularity grows.  \ByteX avoids this becoming a migration event by making storage placement a per-segment property controlled by buffer pool sizing.  

\para{Lesson 5: Compose retrieval operators in one plan.}
In production, it is common that vector similarity is composed with structured predicates, lexical matching, and reranking (Customers A-F).  Executing these operators in one distributed plan avoids materializing large intermediate results and enables end-to-end scheduling that a multi-service pipeline cannot provide.

\section{Related Work}
\label{sec:related}
We briefly review core methods and system designs for vector systems. For comprehensive surveys, we refer the reader to \cite{song2026vector, DBLP:journals/vldb/PanWL24, DBLP:journals/debu/TianYZ0Z023}.

\vspace{0.2em}
\para{Vector Search Algorithms.}
Existing ANN methods can be broadly categorized into three groups: LSH-based methods \cite{DBLP:conf/vldb/GionisIM99,DBLP:journals/pvldb/HuangFZFN15,DBLP:conf/icde/TianZZ22,DBLP:journals/pvldb/ZhengZWHLJ20}, graph-based methods \cite{DBLP:conf/www/DongCL11,DBLP:journals/pami/MalkovY20,DBLP:journals/pvldb/FuXWC19,DBLP:journals/pvldb/ZhaoTHZZ23,DBLP:conf/nips/SubramanyaDSKK19}, and quantization-based methods \cite{DBLP:journals/pami/JegouDS11,DBLP:journals/pacmmod/GaoL24}.
In modern vector databases, graph-based methods such as HNSW \cite{DBLP:journals/pami/MalkovY20} and DiskANN \cite{DBLP:conf/nips/SubramanyaDSKK19} are highly effective for ANN search. They build approximate proximity graphs that capture neighborhood relationships among data points, enabling accurate and efficient query processing. Nevertheless, their high construction cost and substantial random memory/disk accesses often motivate the use of quantization-based methods, such as PQ \cite{DBLP:journals/pami/JegouDS11} and RaBitQ \cite{DBLP:journals/pacmmod/GaoL24}, to reduce computation and memory footprint.

\vspace{0.2em}
\para{Architectures for Vector Data Management.}
Production-grade vector search is increasingly closer to a full data management system than to a standalone ANN operator~\cite{faiss}. Beyond similarity search, systems are expected to provide scalable storage and layout management, durability and recovery~\cite{mohan1992aries}, distributed execution~\cite{van2017distributed}, buffer management~\cite{do2011turbocharging}, and high-concurrency updates~\cite{o1996log,DBLP:conf/sosp/XuLLXCZLYYYCY23}. Vector similarity is also more useful when exposed as a composable operator (e.g., via DSL~\cite{elasticsearch_dsl}) to enable hybrid retrieval (text+vector), structured filtering, and reranking within a unified execution engine.
To meet these system-level requirements, existing solutions can be broadly categorized into three architectural lines. Native vector databases, such as Milvus~\cite{DBLP:conf/sigmod/WangYGJXLWGLXYY21} and Pinecone~\cite{pinecone}, are designed for high-throughput, low-latency similarity search with elastic distributed scaling; their primary focus is vector-first index lifecycle management and distributed ANN execution, including segment organization, background compaction, and parallel search across shards. Relational databases extended with vector support, including pgvector~\cite{DBLP:conf/sigmod/YangLFW20}, GaussDB-Vector~\cite{sun2025gaussdb}, Blendhouse~\cite{DBLP:conf/icde/NiuTPC25}, Cosmos~\cite{DBLP:journals/pvldb/UpretiSSSBPACSYHDPHDMXD25}, and MariaDB-Vector~\cite{mariaDB}, embed vector data types and similarity operators into the relational stack; their central axis is tight integration with the storage manager, execution engine, and cost-based optimizer to enable composable plans with joins and predicates while controlling storage and computation overhead. Enterprise search engines such as OpenSearch~\cite{opensearch2024} and Elasticsearch~\cite{elasticsearch2024} incorporate vector retrieval into text-centric search pipelines; they emphasize the co-design of inverted and vector indexes to support hybrid ranking, filtering, and large-scale distributed search serving.

\section{Conclusion}
\label{sec:conclusion}

Modern AI applications have transformed search from a standalone nearest-neighbor operator into a production data-management problem: vector collections scale rapidly, queries mix vector similarity with diverse constraints, deployment modes shift between memory and disk as costs change, and operators still need mature search APIs and tooling. This paper presented \ByteX, an OpenSearch-compatible AI search engine built from these observations. By combining a high-performance quantization-aware vector kernel, asynchronous SSD-resident graph traversal, and adaptive storage, \ByteX supports both PB-scale ingestion-heavy deployments and latency-sensitive online serving. Our experience shows that modern AI search systems must treat vectors as first-class database primitives while preserving the operational ecosystem that production teams already depend on.



\bibliographystyle{ACM-Reference-Format}
\bibliography{sample}

\appendix

\section*{Appendix}

\section{Ingestion side in \ByteX{}}
\label{sec:appendix_write}
\para{Write and indexing path.}
The write path follows a staged lifecycle that makes fresh data searchable quickly while gradually turning it into serving-efficient indexes. \textcircled{\scriptsize 1} Data first enters through the standard OpenSearch ingestion interface; on refresh, new vectors are exposed as small fresh segments, avoiding heavyweight index construction on the critical ingestion path. \textcircled{\scriptsize 2} As ingestion continues, background merge compacts these segments to bound query fanout, which is essential for sustaining low latency under continuous writes. \textcircled{\scriptsize 3} During this process, \ByteX{} upgrades the index representation with segment maturity: small segments are searched by direct scan, medium segments use IVF, and large stable segments are converted to graph indexes. \textcircled{\scriptsize 4} For graph-indexed segments, construction and merge are performed in the quantized space using \textsf{SymRaBitQ}, eliminating repeated full-precision vector materialization on the recurring merge path. \textcircled{\scriptsize 5} Once merged, the same segment can be placed in memory-resident, hybrid memory--SSD, or SSD-resident serving mode without reindexing or application migration.

\section{Serving side in \ByteX{}}
\label{sec:appendix_read}
\para{Read and serving path.}
The read path retains OpenSearch-style query semantics while hiding the heterogeneity of vector placement and execution. \textcircled{\scriptsize 1} A native DSL query is first parsed by the search coordinator and routed to the relevant shards. \textcircled{\scriptsize 2} Within each distributed query plan, \ByteX{} jointly executes BM25 retrieval, vector similarity search, structured predicates, score fusion, and reranking, so applications do not need to issue separate requests or merge results across services. \textcircled{\scriptsize 3} Vector serving then adapts to placement: memory-resident segments use in-memory graph traversal, whereas SSD-resident graph records are accessed through a record-level buffer pool and coroutine-based asynchronous traversal that overlaps SSD reads with distance computation. \textcircled{\scriptsize 4} Output-heavy queries, including large top-K and radial-style queries, are handled as increased serving pressure within the same execution framework rather than as a separate retrieval mode.
\begin{listing}[t]

\caption{Query DSL of Hybrid Search in \ByteX}

\label{lst:appendix_hybrid_query}


\begin{minted}{text}
|\codehll{GET index\_name/\_search?search\_pipeline=pipeline\_name}|
|\codecomment{// specify the dataset and search pipeline}| 
{ |\codehl{"size"}|: 5,  |\codecomment{// number of returned results}|
  |\codehl{"query"}|: {
    |\codehl{"hybrid"}|: {
      |\codehl{"queries"}|: [
        {|\codehl{"match\_phrase"}|: {|\codehl{"content"}|: "test"}},
          |\codecomment{// Traditional full-text search}| 
        {|\codehl{"knn"}|: {|\codehl{"vector\_field"}|: {|\codehl{"vector"}|: [-0.21, ..., 0.89], |\codehl{"k"}|: 100, |\codehl{"filter"}|:  {conditions}}}}
        |\codecomment{// Vector similarity search with structured predicates}|
            ]
        }
    }
}
\end{minted}
\end{listing}

\section{OpenSearch-Compatible Hybrid Query}
\label{sec:appendix_dsl}

Example~\ref{lst:appendix_hybrid_query} presents an example hybrid query in the native DSL.
The query combines structured predicates, lexical matching, vector top-$k$, score fusion, and a reranking pipeline without requiring applications to call separate retrieval services.

\section{RaBitQ, Graph, and IVF Background}
\label{sec:appendix_rabitq_graph_ivf}

This subsection gives the background that Sec.~\ref{sec:vector_kernel} intentionally keeps short.  The main text focuses on \textsf{SymRaBitQ} and its integration into \ByteX; here we spell out the baseline quantization and indexing mechanics.

\subsection{RaBitQ and ExtRaBitQ}
\label{sec:appendix_rabitq}

RaBitQ~\cite{DBLP:journals/pacmmod/GaoL24} and ExtRaBitQ~\cite{gao2025practical} estimate distances by reducing raw-vector scoring to inner-product estimation between normalized residual vectors.  Given a centroid $\bm{c}$, a raw data vector $\bm{x}_r$, and a raw query vector $\bm{q}_r$, define
\begin{equation}\small
\bm{o}=\frac{\bm{x}_r-\bm{c}}{\|\bm{x}_r-\bm{c}\|},
\qquad
\bm{q}=\frac{\bm{q}_r-\bm{c}}{\|\bm{q}_r-\bm{c}\|}.
\end{equation}
Then the squared Euclidean distance can be written as
\begin{equation}\small
\|\bm{x}_r-\bm{q}_r\|^2
=\|\bm{x}_r-\bm{c}\|^2+\|\bm{q}_r-\bm{c}\|^2
-2\|\bm{x}_r-\bm{c}\|\|\bm{q}_r-\bm{c}\|\langle\bm{o},\bm{q}\rangle,
\end{equation}
and inner-product scoring can be recovered by
\begin{equation}\small
\langle \bm{x}_r,\bm{q}_r\rangle
=\langle \bm{x}_r,\bm{c}\rangle-\|\bm{c}\|^2+\langle \bm{q}_r,\bm{c}\rangle
+\|\bm{x}_r-\bm{c}\|\|\bm{q}_r-\bm{c}\|\langle\bm{o},\bm{q}\rangle.
\end{equation}
Terms involving the stored data vector and centroid are precomputed, so the online problem is to estimate $\langle\bm{o},\bm{q}\rangle$.

RaBitQ builds a conceptual rotated codebook.  Let
\begin{equation}\small
\begin{aligned}
\mathcal{G} &:= \left\{-\frac{2^B-1}{2}+u ~\middle|~ u=0,1,\ldots,2^B-1\right\}^D,\\
\mathcal{G}_r &:= \left\{P\frac{\bm{y}}{\|\bm{y}\|} ~\middle|~ \bm{y}\in\mathcal{G}\right\},
\end{aligned}
\end{equation}
where $P$ is a random orthogonal transformation matrix and $B$ is the number of bits per dimension.  A vector is quantized by finding the nearest element of $\mathcal{G}_r$.  Because explicitly materializing $\mathcal{G}_r$ would require $2^{BD}$ candidates, RaBitQ applies the inverse rotation to the vector and searches in the unrotated integer lattice:
\begin{equation}\small
\argmin_{\bar{\bm{o}}\in\mathcal{G}_r}\|\bar{\bm{o}}-\bm{o}\|^2
=
\argmax_{\bar{\bm{y}}\in\mathcal{G}}
\left\langle\frac{\bar{\bm{y}}}{\|\bar{\bm{y}}\|},P^{-1}\bm{o}\right\rangle.
\end{equation}
This yields a compact code $\bm{u}\in\{0,\ldots,2^B-1\}^D$ and a quantized unit vector $\bar{\bm{o}}$.

For two encoded vectors with integer codes $\bm{u}_o$ and $\bm{u}_q$, define the centered code vectors
$\tilde{\bm{u}}_o=\bm{u}_o-(2^B-1)\bm{1}_D/2$ and
$\tilde{\bm{u}}_q=\bm{u}_q-(2^B-1)\bm{1}_D/2$.
Since $\bar{\bm{o}}=P\tilde{\bm{u}}_o/\|\tilde{\bm{u}}_o\|$, $\bar{\bm{q}}=P\tilde{\bm{u}}_q/\|\tilde{\bm{u}}_q\|$, and $P^TP=I$, the code-only numerator used by \textsf{SymRaBitQ} is
\begin{equation}\small
\label{eq:appendix_sym_code_dot}
\langle \bar{\bm{o}},\bar{\bm{q}}\rangle
=
\frac{\langle \tilde{\bm{u}}_o,\tilde{\bm{u}}_q\rangle}
{\|\tilde{\bm{u}}_o\|\,\|\tilde{\bm{u}}_q\|}.
\end{equation}
For $B=1$, the centered codes have only two signs, so Eq.~(\ref{eq:appendix_sym_code_dot}) is equivalent to normalized sign agreement and can be computed from the Hamming distance between binary codes.

The original RaBitQ/ExtRaBitQ query estimator is asymmetric:
\begin{equation}\small
\label{eq:appendix_rabitq_est}
\widehat{\langle\bm{o},\bm{q}\rangle}_{\mathrm{RaBitQ}}
=\frac{\langle\bar{\bm{o}},\bm{q}\rangle}{\langle\bar{\bm{o}},\bm{o}\rangle}.
\end{equation}
The denominator is a per-vector constant, but the numerator still contains the raw query/residual vector.  This is suitable for data-to-query scoring.  It is not sufficient for graph construction and segment merge, where the hot loop compares one stored vector against other stored candidate vectors.  \textsf{SymRaBitQ} keeps the same code construction but replaces Eq.~(\ref{eq:appendix_rabitq_est}) with the symmetric estimator in Eq.~(\ref{eq:symrabitq_est}), so stored-vector comparisons use only codes and precomputed constants.

\subsection{Graph Index Background}
\label{sec:appendix_graph_background}

Graph-based ANN indexes represent each vector as a vertex and store a bounded set of outgoing neighbors.  A query starts from one or more entry points and performs a greedy search: it keeps a candidate list, repeatedly expands the currently closest unvisited vertex, inserts its outgoing neighbors, and truncates the list to a configured search width.  The quality of the graph determines whether this local greedy process quickly reaches the true nearest-neighbor region.

\ByteX follows the DiskANN/Vamana-style construction process~\cite{DBLP:conf/nips/SubramanyaDSKK19}.  For an inserted vector, the builder first runs a greedy search on the current graph to collect a candidate set.  It then applies \textsf{RobustPrune}, which keeps at most $R$ outgoing neighbors while removing candidates that are close to an already selected neighbor.  In a typical pruning step, after selecting the current closest candidate $\bm{p}$ to the inserted vector $\bm{o}$, another candidate $\bm{v}$ can be removed when $\alpha\,d(\bm{p},\bm{v})\le d(\bm{o},\bm{v})$.  This rule encourages sparse, diverse neighborhoods and approximates the sparse-neighborhood property used by proximity graphs.  The builder also inserts backward edges; if a neighbor's out-degree exceeds $R$, it is pruned again.  Recent theoretical analysis further shows why Vamana-style graphs are attractive among proximity graph families~\cite{DBLP:journals/corr/abs-2510-05975}.

In \ByteX, graph building is additionally segment-aware.  The system computes cluster centers on a sample, assigns each vector to its nearest centers, builds a subgraph inside each cluster, and merges the subgraphs by unioning edges.  This keeps construction parallel and local while allowing large merged segments to obtain a high-quality graph.  The main-text contribution is that every distance $d(\cdot,\cdot)$ in greedy search, pruning, and backward-edge repair is replaced by $\mathsf{SymDist}_B$.

\begin{algorithm}[t]
\small
\caption{SymRaBitQ Graph Construction}
\label{alg:index}
\LinesNumbered
\KwIn{Dataset $\mathcal{D}$, parameters $m,l,R,\alpha,L$}
\KwOut{Graph $G$}

Compute cluster centers $\{\bm{c}_i\}_{i=1}^m$ on a sample of $\mathcal{D}$\;
Assign each $\bm{x}_r\in\mathcal{D}$ to its $l$ nearest centers\;
For each assignment to center $\bm{c}_i$, encode $\bm{x}_r$ as $z_i(\bm{o})=(\bar{\bm{o}},\bar{\bm{o}}_0,\gamma_{\bm{o}},\rho_{\bm{o}})$ via \textsf{SymRaBitQ}\;

\For{$1\le i\le m$}{
  Initialize a random $R$-regular directed graph $G_i$\;
  Choose an entry point $s_i$ for cluster $i$\;
  Randomly permute encoded records in cluster $i$\;

  \ForEach{$z_i(\bm{o})$ \textbf{in parallel}}{
    $\mathcal{V}\gets\textsf{FindCandidates}(G_i,s_i,z_i(\bm{o}),L)$\;
    $N_{\mathrm{out}}(z_i(\bm{o}))\gets\textsf{RobustPrune}(G_i,z_i(\bm{o}),\mathcal{V},\alpha,R)$\;
    \ForEach{$z_j\in N_{\mathrm{out}}(z_i(\bm{o}))$}{
      \If{$|N_{\mathrm{out}}(z_j)\cup\{z_i(\bm{o})\}|>R$}{
        $N_{\mathrm{out}}(z_j)\gets\textsf{RobustPrune}(G_i,z_j,N_{\mathrm{out}}(z_j)\cup\{z_i(\bm{o})\},\alpha,R)$\;
      }
      \Else{
        $N_{\mathrm{out}}(z_j)\gets N_{\mathrm{out}}(z_j)\cup\{z_i(\bm{o})\}$\;
      }
    }
  }
}
Merge all subgraphs $\{G_i\}$ into $G$\;
\Return{$G$}
\end{algorithm}

\begin{algorithm}[t]
\small
\caption{\textsf{FindCandidates}$(G_i,s_i,z_{\bm{o}},L)$}
\label{alg:greedy}
\LinesNumbered
\KwIn{Graph $G_i$, entry point $s_i$, encoded target $z_{\bm{o}}$, search list size $L$}
\KwOut{Visited candidate set $\mathcal{V}$}

$\mathcal{L}\leftarrow\{s_i\}$, $\mathcal{V}\leftarrow\emptyset$\;
\While{$\mathcal{L}\setminus\mathcal{V}\neq\emptyset$}{
  $z^*\gets\argmin_{z_j\in\mathcal{L}\setminus\mathcal{V}}\operatorname{SymDist}_{B}(z_j,z_{\bm{o}})$\;
  $\mathcal{L}\leftarrow\mathcal{L}\cup N_{\mathrm{out}}(z^*)$\;
  $\mathcal{V}\leftarrow\mathcal{V}\cup\{z^*\}$\;
  \If{$|\mathcal{L}|>L$}{
    Keep the closest $L$ records in $\mathcal{L}$ under $\operatorname{SymDist}_{B}(\cdot,z_{\bm{o}})$\;
  }
}
\Return{$\mathcal{V}$}
\end{algorithm}

\begin{algorithm}[t]
\small
\caption{\textsf{RobustPrune}$(G_i,z_{\bm{o}},\mathcal{V},\alpha,R)$}
\label{alg:robustprune}
\LinesNumbered
\KwIn{Graph $G_i$, encoded target $z_{\bm{o}}$, candidate set $\mathcal{V}$, pruning threshold $\alpha$, degree bound $R$}
\KwOut{Updated out-neighbor set $N_{\mathrm{out}}(z_{\bm{o}})$}

$\mathcal{V}\leftarrow(\mathcal{V}\cup N_{\mathrm{out}}(z_{\bm{o}}))\setminus\{z_{\bm{o}}\}$\;
$N_{\mathrm{out}}(z_{\bm{o}})\leftarrow\emptyset$\;
\While{$\mathcal{V}\neq\emptyset$}{
  $z^*\gets\argmin_{z_j\in\mathcal{V}}\operatorname{SymDist}_{B}(z_{\bm{o}},z_j)$\;
  $N_{\mathrm{out}}(z_{\bm{o}})\leftarrow N_{\mathrm{out}}(z_{\bm{o}})\cup\{z^*\}$\;
  \If{$|N_{\mathrm{out}}(z_{\bm{o}})|=R$}{
    \textsf{break}\;
  }
  \ForEach{$z_j\in\mathcal{V}$}{
    \If{$\alpha\cdot\operatorname{SymDist}_{B}(z^*,z_j)\le\operatorname{SymDist}_{B}(z_{\bm{o}},z_j)$}{
      Remove $z_j$ from $\mathcal{V}$\;
    }
  }
}
\Return{$N_{\mathrm{out}}(z_{\bm{o}})$}
\end{algorithm}

\begin{algorithm}[t]
\small
\caption{SymRaBitQ Graph Search}
\label{alg:graph-search}
\LinesNumbered
\KwIn{Graph $G$, entry point $s$, query $\bm{q}_r$, parameters $k,L,\beta$}
\KwOut{Top-$k$ nearest neighbors}

Encode $\bm{q}_r$ as $z_{\bm{q}}=(\bar{\bm{q}},\bar{\bm{q}}_0,\gamma_{\bm{q}},\rho_{\bm{q}})$ via \textsf{SymRaBitQ}\;
$\mathcal{L}\leftarrow\{s\}$, $\mathcal{V}\leftarrow\emptyset$\;
\While{$\mathcal{L}\setminus\mathcal{V}\neq\emptyset$}{
  $z^*\gets\argmin_{z_j\in\mathcal{L}\setminus\mathcal{V}}\operatorname{SymDist}_{1}(z_j,z_{\bm{q}})$\;
  $\mathcal{L}\leftarrow\mathcal{L}\cup N_{\mathrm{out}}(z^*)$\;
  $\mathcal{V}\leftarrow\mathcal{V}\cup\{z^*\}$\;
  \If{$|\mathcal{L}|>L$}{
    Keep the closest $L$ records in $\mathcal{L}$ under $\operatorname{SymDist}_{1}(\cdot,z_{\bm{q}})$\;
  }
}
$\mathcal{S}\leftarrow$ closest $\beta k$ records in $\mathcal{L}$ under $\operatorname{SymDist}_{1}(\cdot,z_{\bm{q}})$\;
\Return{closest $k$ records in $\mathcal{S}$ under $\operatorname{SymDist}_{B}(\cdot,z_{\bm{q}})$}
\end{algorithm}

\subsection{IVF and FastScan Background}
\label{sec:appendix_ivf_background}
An inverted-file (IVF) index partitions vectors by KMeans centroids.  Each centroid owns an inverted list containing the vectors assigned to it, usually represented as residuals relative to the centroid.  At query time, the system first compares the query with all centroids, chooses the closest $n_{\mathit{probe}}$ lists, and scans only those lists.  IVF is cheap to build and works well for medium-size segments because it reduces scan volume without paying the heavier graph-construction cost.

IVF is commonly paired with quantization.  The query is encoded or transformed with respect to each probed centroid, and list entries are scored with compressed codes.  FastScan-style methods~\cite{DBLP:journals/pvldb/AndreKS15,DBLP:conf/mir/AndreKS17,DBLP:journals/pami/AndreKS21} speed up this scan by preparing query-specific lookup tables and evaluating low-bit codes with SIMD-friendly operations.  \ByteX uses the same idea in a SymRaBitQ form: the first-bit code gives a very cheap coarse distance estimate, and the remaining bits are used only for candidates that survive coarse filtering.

\section{Proof of Theorem~\ref{th:estimator}}
\label{sec:theorem_4.1_proof}
\begin{proof}
We first show that the mean depends only on the inner product, and
then prove concentration around this mean.
Fix the unit input vectors $\bm{o},\bm{q}$ and let $P$ be a shared uniformly random orthogonal matrix, independent of the inputs. For one-bit RaBitQ, define
\begin{equation}
\small
 \bar{\bm{x}}(P)=\frac{P\operatorname{sign}(P^{\mathsf T}\bm{x})}{\sqrt D},
 \qquad \gamma_{\bm{x}}=\frac{\|P^{\mathsf T}\bm{x}\|_1}{\sqrt D}.
\end{equation}
Thus $D^{-1/2}\le\gamma_{\bm{x}}\le1$. Any two unit-vector pairs
with the same inner product are related by an orthogonal transformation
$U$. Since $UP$ has the same distribution as $P$ and
$\overline{U\bm{x}}(UP)=U\bar{\bm{x}}(P)$, the estimator has the same
distribution, and hence expectation, for both pairs. Thus the mean
depends only on the inner product. Inversion additionally requires
this response to be one-to-one on the interval of interest.

For the concentration bound, fix $t>0$ and write
$N=\langle\bar{\bm{o}},\bar{\bm{q}}\rangle$, $n=\mathbb{E}N$, and
\begin{equation}
\small
 \mu=\mathbb{E}\gamma_{\bm{x}}
 =\sqrt{\frac D\pi}\frac{\Gamma(D/2)}{\Gamma((D+1)/2)}
 \ge\sqrt{2/\pi}.   
\end{equation}

We begin by bounding the fluctuation of the numerator $N$ around $n$.
Let $\theta=\arccos\langle\bm{o},\bm{q}\rangle$.
For $0<\theta<\pi$, represent the input pair in its two-dimensional
span by $e_1=(1,0)^{\mathsf T}$ and
$v=(\cos\theta,\sin\theta)^{\mathsf T}$.
Let $G\in\mathbb{R}^{D\times2}$ have independent standard Gaussian
entries and take its thin QR decomposition $G=QR$, where
$Q\in\mathbb{R}^{D\times2}$ satisfies $Q^{\mathsf T}Q=I_2$, and
$R\in\mathbb{R}^{2\times2}$ is upper triangular with positive diagonal.
Gaussian rotational invariance makes the columns of $Q$ a uniformly
random orthonormal pair. Thus $(Qe_1,Qv)$ has the same distribution
as $(P^{\mathsf T}\bm{o},P^{\mathsf T}\bm{q})$.
Moreover, $Qe_1=Ge_1/R_{11}$ and $Qv=GR^{-1}v$.
For $w=R^{-1}v$, we have $w_2>0$ and, almost surely,
\begin{equation}
\small
 N=1-2F_D(-w_1/w_2),
\end{equation}
where $F_D$ and $F$ are the empirical and common distribution functions
of the independent Cauchy variables $G_{i2}/G_{i1}$.
The DKW--Massart inequality~\cite{massart1990dkw} gives
\begin{equation}
\small 
 \mathbb{P}\!\left\{\sup_z|F_D(z)-F(z)|>
 \frac{t}{2\sqrt D}\right\}\le2e^{-t^2/2}.
\end{equation}
The uniform bound applies even though $-w_1/w_2$ depends on $G$.
Let $\theta_G$ be the angle between $R^{-1}e_1$ and $R^{-1}v$.
The random-hyperplane identity gives
$1-2F(-w_1/w_2)=1-2\theta_G/\pi$ and $n=1-2\theta/\pi$.
Except with probability $2e^{-t^2/2}$,
\cite[Corollary~5.35]{vershynin2012randommatrices} gives
\begin{equation}
\small
 \sqrt D-\sqrt2-t\le\sigma_{\min}(G)
 \le\sigma_{\max}(G)\le\sqrt D+\sqrt2+t. 
\end{equation}
For $r=(\sqrt2+t)/\sqrt D<1$, the condition number of $R^{-1}$
is at most $(1+r)/(1-r)$, so
$|\theta_G-\theta|\le2\arcsin r\le\pi r$.
Indeed, polar decomposition reduces the angle distortion to a
positive definite matrix $S$, for which
\begin{equation}
\small
 \sin\angle(Su,u)\le
 \frac{\lambda_{\max}(S)-\lambda_{\min}(S)}
 {\lambda_{\max}(S)+\lambda_{\min}(S)}.
\end{equation}
Applying this to both vectors bounds the change of their angle.
Combining the two events gives
$|N-n|\le(3t+2\sqrt2)/\sqrt D$.
If $r\ge1$, this inequality follows from $|N-n|\le2$.
For collinear inputs, $N$ is constant. Therefore
\begin{equation}\label{eq:appendix_numerator_tail}
\small
 \mathbb{P}\!\left\{|N-n|>\frac{3t+2\sqrt2}{\sqrt D}\right\}
 \le4e^{-t^2/2}.
\end{equation}

Next, we show that both self-correlations concentrate around $\mu$.
Set $s=t/\sqrt{D-2}$, the deviation threshold for the self-correlations.
The function $u\mapsto\|u\|_1/\sqrt D$ is $1$-Lipschitz on the
unit sphere. The sphere's log-Sobolev constant is at most $1/(D-2)$,
so the Herbst argument~\cite[Theorems~3.5 and~3.8]{meckes2014concentration}
gives
\begin{equation}\label{eq:appendix_denominator_tail}
\small
 \mathbb{P}\{|\gamma_{\bm{x}}-\mu|>s\}\le2e^{-t^2/2}.
\end{equation}
\begin{table*}[!t]
\small
\centering
\caption{Controlled benchmark resource and index configurations. C denotes vCPUs.}
\label{tab:cluster_config}
\renewcommand{\arraystretch}{1.25}
\begin{tabular}{l | p{0.13\linewidth} | p{0.25\linewidth} | p{0.33\linewidth}}
\toprule
System & Datasets & Resource configuration & Index and tuning configuration \\
\midrule
\multirow{2}{*}{ElasticSearch}
& Wiki, GIST
& 3 $\times$ 8C64GB data nodes \newline 2 $\times$ 8C16GB coordinator nodes
& $m=48$, \texttt{ef\_construction}=400 \\
\cmidrule{2-4}
& Cohere, in-house
& 3 $\times$ 24C192GB data nodes \newline 2 $\times$ 16C32GB coordinator nodes
& $m=48$, \texttt{ef\_construction}=400 for Cohere\newline $m=20$, \texttt{ef\_construction}=400 for in-house \\
\midrule
\multirow{2}{*}{\ByteX}
& Wiki, GIST
& 3 $\times$ 8C64GB data nodes \newline 2 $\times$ 8C16GB coordinator nodes
& $m=64$, \texttt{ef\_construction}=300 \\
\cmidrule{2-4}
& Cohere, in-house
& 3 $\times$ 24C192GB data nodes \newline 2 $\times$ 16C32GB coordinator nodes
& $m=48$, \texttt{ef\_construction}=200 \\
\midrule
\multirow{2}{*}{\makecell{PostgreSQL\\(VectorChord)}}
& Wiki, GIST
& Single large host capped by cgroups at 24C192GB
& $m=20$, \texttt{ef\_construction}=400 \\
\cmidrule{2-4}
& Cohere, in-house
& Single large host capped by cgroups at 72C576GB
& $m=48$, \texttt{ef\_construction}=400 \\
\midrule
\multirow{2}{*}{Milvus}
& Wiki, GIST
& 3 $\times$ 8C64GB query nodes \newline 3 $\times$ 8C64GB index nodes \newline 2 $\times$ 8C16GB proxy nodes \newline 2 $\times$ 8C16GB data nodes \newline 2 $\times$ 4C8GB mixcoord nodes
& $M=24$, \texttt{ef\_construction}=300 \\
\cmidrule{2-4}
& Cohere, in-house
& 3 $\times$ 24C192GB query nodes\newline 3 $\times$ 24C192GB index nodes\newline 2 $\times$ 16C32GB proxy nodes\newline 2 $\times$ 16C32GB data nodes\newline 2 $\times$ 8C16GB mixcoord nodes
& $M=48$, \texttt{ef\_construction}=400 \\
\bottomrule
\end{tabular}
\end{table*}

Then, we prove that the estimator's expectation differs from
$n/\mu^2$ by at most $20/D$.
For a uniform unit vector $u$, put
\begin{equation}
\small
 h(u)=\frac{D}{\|u\|_1^2},\qquad T=D\sum_i u_i^4.
\end{equation}
Norm interpolation gives $1\le h\le T$. Standard sphere moments give
\begin{equation}
\small
 \begin{aligned}
 &\mathbb{E}T=\frac{3D}{D+2}\le3,\\
 &\mathbb{E}(T^2-T)=
 \frac{6D(D-1)(D+12)}{\prod_{j=1}^3(D+2j)}\le6,\\
 &\mathbb{E}(T^3-T^2)=
 \frac{6D^2(D-1)(3D^2+104D+1280)}
 {\prod_{j=1}^5(D+2j)}\le24.
 \end{aligned}
\end{equation}
These follow by expanding the powers and using the coordinate moments
of a normalized standard Gaussian vector.

Let $U_{\bm{x}}=\gamma_{\bm{x}}^{-1}$, $m=\mathbb{E}U_{\bm{x}}$,
and $v_0=(1-2/\pi)/2$. The squared spherical gradients of
$\gamma_{\bm{x}}$, $U_{\bm{x}}$, and $\gamma_{\bm{x}}^{-2}$
are, respectively, $1-\gamma_{\bm{x}}^2$, $h^2-h$, and
$4(h^3-h^2)$. Applying the even-function spherical
Poincar\'e inequality~\cite{bianchini2023poincare}
by smooth approximation gives
\begin{equation}
\small
 \begin{aligned}
 \operatorname{Var}(\gamma_{\bm{x}})&\le v_0/D,\\
 \operatorname{Var}(U_{\bm{x}})&\le3/D,\qquad
 \operatorname{Var}(\gamma_{\bm{x}}^{-2})\le48/D.
 \end{aligned}
\end{equation}
Also $m\le\sqrt{\mathbb{E}h}\le\sqrt3$, and
$\gamma_{\bm{x}}U_{\bm{x}}=1$ implies
\begin{equation}
\small
 0\le\mu m-1
 =-\operatorname{Cov}(\gamma_{\bm{x}},U_{\bm{x}})
 \le\frac{\sqrt{3v_0}}D.
\end{equation}
Write $A=U_{\bm{o}}U_{\bm{q}}$ and $\nu=\mathbb{E}h$.
Since $A$ is the geometric mean of
$\gamma_{\bm{o}}^{-2}$ and $\gamma_{\bm{q}}^{-2}$,
\begin{equation}
\small
 (A-\nu)^2\le
 (\gamma_{\bm{o}}^{-2}-\nu)^2+
 (\gamma_{\bm{q}}^{-2}-\nu)^2,
\end{equation}
so $\operatorname{Var}(A)\le96/D$.
The covariance inequality for an orthonormal
basis~\cite[Theorem~2]{goldstein2017orthonormal} gives
$\operatorname{Var}(N)\le2/D$.
Using $\mathbb{E}A=m^2+\operatorname{Cov}(U_{\bm{o}},U_{\bm{q}})$
and $|n|\le1$, we obtain
\begin{equation}\label{eq:appendix_center_shift}
\small
 \begin{aligned}
 \left|\mathbb{E}[\widehat{\langle\bm{o},\bm{q}\rangle}]
       -\frac{n}{\mu^2}\right|
 &\quad\le
 \sqrt{\operatorname{Var}(N)\operatorname{Var}(A)}
 +\operatorname{Var}(U_{\bm{x}})\\
 &\qquad +(m+\mu^{-1})(m-\mu^{-1})
 <\frac{20}{D}.
 \end{aligned}
\end{equation}
Indeed, the coefficient is at most
\begin{equation}
\small
 8\sqrt3+3+(\sqrt3+\sqrt{\pi/2})\sqrt{3\pi v_0/2}<20.
\end{equation}
No independence between the numerator and denominators is required.

Finally, we combine these bounds to obtain the claimed concentration
inequality.
Outside events of total probability at most $8e^{-t^2/2}$,
Eqs.~\eqref{eq:appendix_numerator_tail} and
\eqref{eq:appendix_denominator_tail} hold simultaneously.
Since $\gamma_{\bm{x}}\le1$,
\begin{equation}
\small
 \begin{aligned}
 |\gamma_{\bm{o}}\gamma_{\bm{q}}-\mu^2|
 &\le\gamma_{\bm{o}}|\gamma_{\bm{q}}-\mu|
      +\mu|\gamma_{\bm{o}}-\mu|\\
 &\le(1+\mu)s.
 \end{aligned}
\end{equation}
Moreover, $(1+\mu)/\mu^2\le\pi/2+\sqrt{\pi/2}<3$.
Together with Eq.~\eqref{eq:appendix_center_shift}, this yields
\begin{equation}
\small
 \begin{aligned}
 \left|\widehat{\langle\bm{o},\bm{q}\rangle}
 -\mathbb{E}[\widehat{\langle\bm{o},\bm{q}\rangle}]\right|
 &\quad\le\frac{20}{D}
 +\frac{(3t+2\sqrt2)/\sqrt D+(1+\mu)s/\mu^2}
 {\gamma_{\bm{o}}\gamma_{\bm{q}}}\\
 &\quad\le\frac{20}{D}
 +\frac{6t+3}{\gamma_{\bm{o}}\gamma_{\bm{q}}\sqrt{D-2}}.
 \end{aligned}
\end{equation}
For $D\ge400$, the bound $\gamma_{\bm{o}}\gamma_{\bm{q}}\le1$ gives
\begin{equation}
\small
 \frac{20}{D}\le\frac1{\sqrt D}
 \le\frac1{\gamma_{\bm{o}}\gamma_{\bm{q}}\sqrt{D-2}}.
\end{equation}
Absorbing this term into the numerator and taking $t=\epsilon_0$
proves Eq.~\eqref{eq:symrabitq_centered_bound}; thus $D\ge400$
is sufficient for the theorem's large-dimension condition.
On the same event, $\gamma_{\bm{o}}\gamma_{\bm{q}}\ge(\mu-s)^2$
whenever $s<\mu$. In particular,
$t=\sqrt{2\log(8/\delta)}$ gives confidence $1-\delta$ and a
deterministic radius of order
$O(\sqrt{\log(8/\delta)/D})+O(D^{-1})$
when $\log(8/\delta)=o(D)$.
For any fixed confidence level $1-\delta$, this gives the
$O(D^{-1/2})$ form in Eq.~\eqref{eq:symrabitq_concise_bound}.
\end{proof}




\begin{figure*}[t]
    \centering

    \begin{minipage}[t]{\textwidth}
        \centering
        \includegraphics[width=\linewidth]{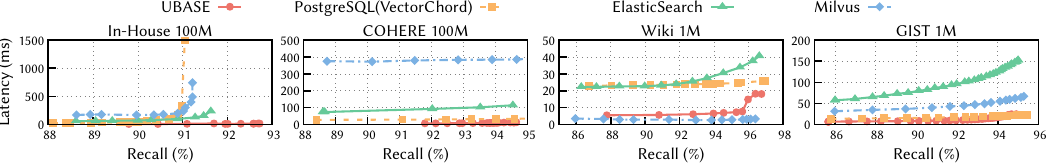}
        \caption{Concurrent Average Latency--Recall Trade-off}
        \label{fig:appendix-latency-recall-conc-avg}
    \end{minipage}

    \begin{minipage}[t]{\textwidth}
        \centering
        \includegraphics[width=\linewidth]{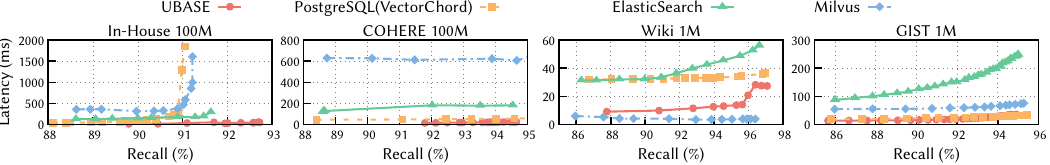}
        \caption{Concurrent P95 Latency--Recall Trade-off}
        \label{fig:appendix-latency-recall-conc-p95}
    \end{minipage}

    \begin{minipage}[t]{\textwidth}
        \centering
        \includegraphics[width=\linewidth]{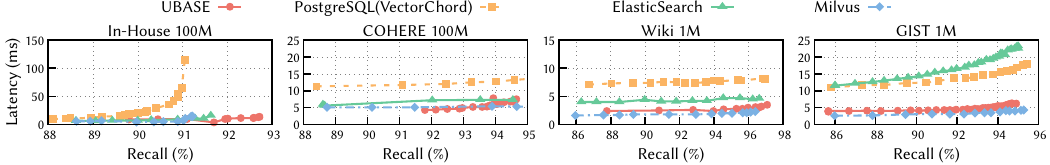}
        \caption{Serial P95 Latency--Recall Trade-off}
        \label{fig:appendix-latency-recall-seri-p95}
    \end{minipage}

    \begin{minipage}[t]{\textwidth}
        \centering
        \includegraphics[width=\linewidth]{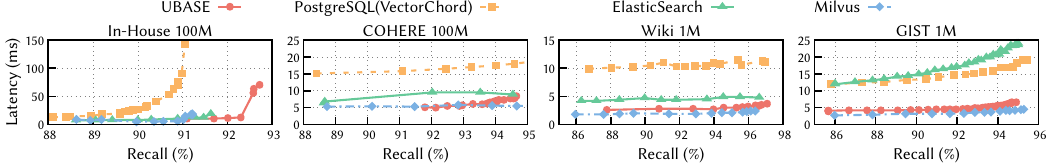}
        \caption{Serial P99 Latency--Recall Trade-off}
        \label{fig:appendix-latency-recall-seri-p99}
    \end{minipage}

\end{figure*}

\vspace{-1em}
\section{Controlled Benchmark Configuration}
\label{sec:machine_settings}
Table~\ref{tab:cluster_config} summarizes the resource provisioning and index configurations used in the controlled benchmark deployments. For Elasticsearch and \ByteX, the coordinator nodes run query coordination only; the data nodes store and serve the vector index. PostgreSQL runs on a large host with cgroup limits to match the aggregate compute and memory budget of the corresponding data tier. Milvus uses its standard role-separated deployment.
Across distributed systems, Wiki and GIST use one shard with three total copies, while Cohere and the in-house dataset use three shards without replicas. PostgreSQL/VectorChord runs as a single-node system. To put each distributed system in its optimized query-serving state, we consolidate each shard into a single segment before evaluation for Elasticsearch, Milvus, and \ByteX. For PostgreSQL, we follow VectorChord's official performance-tuning guidance\footnote{https://docs.vectorchord.ai/vectorchord/usage/performance-tuning.html} and apply the parameters listed below: \texttt{max\_worker\_processes} = 80, \texttt{max\_parallel\_workers} = 71, \texttt{max\_parallel\_maintenance\_workers} = 71, \texttt{effective\_io} \texttt{\_concurrency} = 20, \texttt{maintenance\_io}\texttt{\_concurrency} = 20, \texttt{jit}= off, and
\texttt{shared\_buffers} = 460GB

\section{Operating-Cost Configuration}
\label{sec:cost_deployment_config}

The operating-cost comparison in Table~\ref{tab:price} charges only the serving tier that physically stores or serves the in-house dataset. Auxiliary components, such as coordinators, proxies, index nodes, and mixcoord nodes, are excluded because their overhead is shared or amortized and does not scale directly with dataset size. Milvus, Elasticsearch, \ByteX (Memory-only), and PostgreSQL/VectorChord use the in-house serving footprint shown in Table~\ref{tab:cluster_config}; PostgreSQL/VectorChord is charged by its cgroup-limited serving footprint.\footnote{A larger host may be used for deployment, but the resource and cost model charges only the cgroup-limited serving footprint.}
All charged nodes use the same commercial cloud instance family unless otherwise stated. The disk-only deployments run on 16C32GB data nodes with 700GB SSDs because only the minimal index structures are memory-resident, while the remaining data are stored on disk.


%



\section{Additional Latency Results}
\label{sec:appendix_latency}

Fig.~\ref{fig:appendix-latency-recall-conc-avg}--\ref{fig:appendix-latency-recall-seri-p99} report the remaining latency--recall curves for the same controlled benchmark setting as Sec.~\ref{sec:exp}.

\end{document}